\documentclass{amsart}
\usepackage[utf8]{inputenc}
\usepackage{amssymb,amsmath,amsfonts,amsthm}
\usepackage[foot]{amsaddr}
\usepackage{geometry}
\usepackage{setspace,multicol,multirow,enumerate,textcomp}
\usepackage{graphicx}
\usepackage{verbatim}
\usepackage{url}
\usepackage[dvipsnames]{xcolor}
\usepackage[makeroom]{cancel}
\usepackage{tikz}
\usetikzlibrary{calc}
\usetikzlibrary{positioning,automata,fit,arrows.meta}
\makeatletter
\def\Ddots{\mathinner{\mkern1mu\raise\p@
\vbox{\kern7\p@\hbox{.}}\mkern2mu
\raise4\p@\hbox{.}\mkern2mu\raise7\p@\hbox{.}\mkern1mu}}
\makeatother

\newtheorem{theorem}{Theorem}

\newtheorem{lemma}[theorem]{Lemma}
\newtheorem{proposition}[theorem]{Proposition}

\numberwithin{theorem}{section}
\newtheorem{definition}{Definition}
\newtheorem{problem}[definition]{Problem}
\numberwithin{definition}{section}
\theoremstyle{remark}
\newtheorem{example}{Example}
\numberwithin{example}{section}
\newcommand{\norm}[1]{\lVert#1\rVert}
\renewcommand{\O}{\mathcal{O}}

\title[The Frobenius and factor universality problems]{The Frobenius and factor universality problems\\of the Kleene star of a finite set of words}
\author{Maksymilian Mika}
\email{mika.maksymilian@gmail.com}
\author{Marek Szyku{\l}a}
\email{msz@cs.uni.wroc.pl}
\address{Institute of Computer Science, University of Wroc{\l}aw, Wroc{\l}aw, Poland}
\begin{document}
\begin{abstract}
We solve open problems concerning the Kleene star $L^*$ of a finite set $L$ of words over an alphabet $\Sigma$.
The \emph{Frobenius monoid} problem is the question for a given finite set of words $L$, whether the language $L^*$ is cofinite.
We show that it is PSPACE-complete.
We also exhibit an infinite family of sets $L$ such that the length of the longest words not in $L^*$ (when $L^*$ is cofinite) is exponential in the length of the longest words in $L$ and subexponential in the sum of the lengths of words in $L$.
The \emph{factor universality} problem is the question for a given finite set of words $L$, whether every word over $\Sigma$ is a factor (substring) of some word from $L^*$.
We show that it is also PSPACE-complete.
Besides that, we exhibit an infinite family of sets $L$ such that the length of the shortest words not being a factor of any word in $L^*$ is exponential in the length of the longest words in $L$ and subexponential in the sum of the lengths of words in $L$.
This essentially settles in the negative the longstanding Restivo's conjecture (1981) and its weak variations.
All our solutions base on one shared construction, and as an auxiliary general tool, we introduce the concept of \emph{set rewriting systems}.
Finally, we complement the results with upper bounds.

\end{abstract}
\maketitle
\section{Introduction}

Given a set of words $L$ over a finite alphabet $\Sigma$, the language $L^*$ contains all finite strings built by concatenating any number of words from $L$.
In general, we can think of $L$ as a dictionary and $L^*$ as the language of all available phrases.
One of the most basic questions that one could ask is whether $L$ generates all words over the alphabet $\Sigma$.
The answer is, however, trivial, because this is the case if and only if $L$ contains all single letters $a \in \Sigma$.
Thus, more useful relaxed questions are considered.
In this paper, we consider two classical such problems, settling their computational complexity and solving the related combinatorial questions.

Let $\norm{L}_1$ denote the sum of the lengths of the words in $L$, and let $\norm{L}_\infty$ denote the maximum length of the words in $L$.
The value $\norm{L}_1$ can be treated as the size of the input.
Note that $\norm{L}_1$ can be exponentially larger than $\norm{L}_\infty$.
We consider complexity and bounds in terms of both values.

\subsection{Frobenius monoid problem}

The classical Frobenius problem is, for given positive integers $x_1,\ldots,x_k$, to determine the largest integer $x$ that is not expressible as a non-negative linear combination of them.
An integer $x$ is expressible as a non-negative linear combination if there are integers $c_1,\ldots,c_k \ge 0$ such that $x = c_1 x_1+\ldots+c_k x_k$.
In a decision version of the problem, we ask whether the largest integer exists, i.e., whether the set of non-expressible positive integers is finite.
It is well known that the answer is ``yes'' if and only if $\gcd(x_1,\ldots,x_k)=1$.

The Frobenius problem was extensively studied and found applications across many fields, e.g., to primitive sets of matrices \cite{Dulmage1964}, the Shellsort algorithm \cite{Incerpi1985}, and counting points in polytopes \cite{Beck2002}.
The problem of computing the largest non-expressible integer is NP-hard \cite{Ramirez1996} when the integers are given in binary, and it can be solved polynomially if the number $k$ of given integers is fixed \cite{Kannan1992}.

A generalization of the Frobenius problem to the setting of languages was introduced by Kao, Shallit, and Xu \cite{KaoShallitXu2008FPFM}.
Instead of a finite set of integers, we are given a finite set of words over some finite alphabet $\Sigma$, and instead of multiplication, we have the usual word concatenation.
The original question becomes whether all but a finite number of words can be expressed as a concatenation of the words from the given set.
If $L$ is our given finite language, then the problem is equivalent to deciding whether $L^*$ is cofinite, i.e., the complement of $L^*$ is finite.

\begin{problem}[Frobenius Monoid Problem for a Finite Set of Words]\label{pbm:fpfl}
Given a finite set of words $L$ over a finite non-empty alphabet $\Sigma$, is $L^*$ cofinite?
\end{problem}

It is a simple observation that, if $\Sigma$ is a unary alphabet, then Problem~\ref{pbm:fpfl} is equivalent to the original Frobenius problem on integers, thus it is polynomially solvable.
There are also efficient algorithms for checking whether a \emph{given} word is in $L^*$ \cite{CDGPR2005ParsingWithAFiniteDictionary}.

\begin{example}
The language $L = \{000, 00000\}$ over $\Sigma = \{0\}$ generates the cofinite language
$L^*$; since $\gcd(3,5)=1$, the language $L^*$ includes all words longer than $3\cdot 5 - 3 - 5 = 7$.
\end{example}

\begin{example}\label{ex:fpfl2}
For the language $L = \{0,01,10,11,101\}$ over $\Sigma = \{0,1\}$, the words in $L^*$ are:
$$ 0, 00, 01, 10, 11, 000, 001, 010, 011, 100, 101, 110, \ldots .$$
We can see that $111 \notin L^*$ and also every word of the form $111(11)^*$ does not belong to $L^*$.
However, if we add $111$ to $L$, the answer becomes that $L^*$ is cofinite; it contains every word except the word~$1$.
\end{example}

The problem can be seen as \emph{almost universality} of the language $L^*$.
Kao, Shallit, and Xu \cite{KaoShallitXu2008FPFM,Xu2009FPFM} showed that, in particular, if $L^*$ is cofinite, then the longest non-expressible words can be exponentially long in $\norm{L}_\infty$; their construction is based on the so-called \emph{multi-shift de Bruijn sequences} \cite{KariXu2012}.
This is in contrast with the classical Frobenius problem, where the largest non-expressible integer is bounded quadratically in the largest given integer \cite{Brauer1942}.
A quadratic bound exists also for a similar problem where the iterated shuffle is used instead of the Kleene star operation \cite{NicholsonRampersad2018}.
Since the shown examples also use exponentially many words in $\norm{L}_\infty$, the question about the bound in terms of $\norm{L}_1$ or $|L|$ remained open.

In 2009, Shallit and Xu posed the open question about the computational complexity of determining whether $L^*$ is cofinite \cite{Xu2009FPFM}.
They also noted that it is NP-hard and in PSPACE when $L$ is given as a regular expression \cite{XuShallit2008}.
This question appears on Shallit's list of open problems \cite{ShallitBC4}.

\subsection{Factor universality problem}

A word $u \in \Sigma^*$ is a \emph{factor} (also called \emph{substring}) of a word $w \in \Sigma^*$ if $v u v'=w$ for some words $v,v' \in \Sigma^*$.
A language $K \subseteq \Sigma^*$ such that every word over $\Sigma$ is a factor of some word from $K$ is called \emph{factor universal}.

\begin{problem}[Factor Universality for a Finite Set of Words]\label{pbm:factor_universality}
Given a finite set of words $L$ over a finite non-empty alphabet $\Sigma$, is $L^*$ factor universal?
\end{problem}

Finite sets $L$ such that $L^*$ is factor universal are one of the basic concepts in the theory of codes \cite[Section~1.5]{BerstelPerrinReutenauer2010}.
They are called \emph{complete sets of words}, and words that are factors of some word in $L^*$ are called \emph{completable}.

\begin{example}
The set $L = \{01, 10, 11, 000\}$ over $\Sigma = \{0,1\}$ is not complete, since the word $100010001$ is not completable.
To create a word that contains $1$ surrounded by $0$s, we have to use either $10$ or $01$, but then there is no way to build the succeeding $001$ or preceding $100$, respectively.
\end{example}

\begin{example}
The set $L = \{00, 01, 10, 11\}$ over $\Sigma = \{0,1\}$ is complete, because every binary sequence of even length is in $L^*$.
We can construct every odd-length binary sequence by removing the first letter of a suitable even-length sequence.
\end{example}

The question about the length of the shortest incompletable words was posed in~1981 by Restivo \cite{Restivo1981}, who conjectured that if a finite set $L$ is not complete, then the shortest incompletable words have length at most $2\norm{L}_\infty^2$.
The conjecture in this form turned out to be false \cite{FiciPribavkinaSakarovitch2010} and $5\norm{L}_\infty^2-\O(\norm{L}_\infty)$ was the best lower bound known so far \cite{GusevPribavkina2011}, but the relaxed question whether there is a quadratic, or at least polynomial, upper bound remained open and became one of the longstanding unsolved problems in automata theory and the theory of codes.
It was generally believed that Restivo's conjecture holds with a larger value of the constant \cite{BerstelPerrinReutenauer2010}.
On the other hand, a sophisticated experimental research dedicated just to this problem \cite{JuliaMalapertProvillard2017} suggested that the tight upper bound is unlikely to be quadratic.
The best known upper bound was a trivial one, exponential in $\norm{L}_1$ thus doubly-exponential in terms of $\norm{L}_\infty$ \cite{GusevPribavkina2011}.

A polynomial upper bound $\O(\norm{L}_1^5)$ was recently derived for the class of sets $L$ called \emph{codes}, which guarantees a unique (unambiguous) factorization of any word to words from $L$ \cite{KieferMascle2019}.
Since $\norm{L}_1$ can be exponentially larger than $\norm{L}_\infty$, so the general question about a polynomial upper bound in $\norm{L}_\infty$ for this subclass still remains open.

The computational complexity of Problem~\ref{pbm:factor_universality} was also an open question.
In a more general setting, where instead of checking the factor universality of $L^*$ we check it for an arbitrary regular language specified by an NFA, the problem was shown to be PSPACE-complete \cite{RampersadShallitXu2012}.
In contrast, it is solvable in linear time when the language is specified by a DFA \cite{RampersadShallitXu2012}.
Also, some upper bound on the length of the shortest incompletable words was recently derived for the case where the language is specified by an unambiguous NFA \cite{BoccutoCarpi2019}.

Both the computational complexity question and finding the tight upper bound on the length also appear as one of Berstel, Perrin, and Reutenauer's research problems \cite[Resarch problems]{BerstelPerrinReutenauer2010} and on Shallit's list \cite{ShallitBC4}.
The problem itself has been connected with a number of different problems, e.g., testing if all bi-infinite words can be generated by a given list of finite words \cite{RampersadShallitXu2012}, synchronizing automata and the famous \v{C}ern{\'y} conjecture \cite{CarpiDAlessandro2017}, and the matrix mortality problem \cite{KieferMascle2019}.
In consequence for the mentioned problems, our solution reveals that the testing problem is PSPACE-complete, that any general weak version of Restivo's conjecture cannot be used to derive good upper bounds for synchronization of automata, and that the matrix mortality problem remains hard when the matrices are restricted to a specific form related to a list of words.

\subsection{Contribution}

We show that both Problem~\ref{pbm:fpfl} and Problem~\ref{pbm:factor_universality} are PSPACE-complete.
We also show exponential in $\norm{L}_\infty$ and subexponential in $\norm{L}_1$ lower bounds for the related length questions.
The complexity and the bounds hold even when the alphabet is binary.
Since as the input we take a list of words, this also settles the complexity of all problem variants where $L$ is given as a DFA, a regular expression, or an NFA.

To make the reduction feasible, we construct it in several steps.
We introduce a rewriting system called \emph{set rewriting} (Section~\ref{sec:set_rewriting}), which is a basis for intermediate problems that we reduce from.
We translate a set rewriting system first to a DFA, then to a binary DFA, and finally to a binary list of words.
The solutions for both problems are based on the same construction of the reduction (Section~\ref{sec:fpfl} and~\ref{sec:fu}), with some technical differences.
Thus, it seems that the methods may be applicable to some other problems concerning the Kleene star.

The answer for the Frobenius monoid problem can be surprising because the problem is equally hard when $L$ is represented by other common representations that are exponentially more succinct (i.e.\ DFA, regular expression, or NFA).
Kao et al.~\cite{KaoShallitXu2008FPFM} gave examples of finite languages $L$ such that the longest words not present in the generated cofinite language $L^*$ are of exponential length in $\norm{L}_\infty$.
However, the number of words in $L$ is also exponential in these examples, thus they do not imply a large lower bound in terms of the size of the input $\norm{L}_1$.
We strengthen this result by exhibiting examples such that the longest words not present in cofinite $L^*$ are of subexponential length in $\norm{L}_1$.
The examples are derived from our reduction and its complexity analysis.

The solution for the factor universality problem uses a similar construction.
As well, as a corollary, we exhibit a family of sets $L$ of binary words whose shortest incompletable words are of exponential length in $\norm{L}_\infty$ and subexponential in $\norm{L}_1$.
This settles in the negative all weak variations of Restivo's conjecture and essentially closes the longstanding problem.

Finally (Section~\ref{sec:upper_bounds}), we note that both problems can be solved in exponential time in $\norm{L}_\infty$ while remaining polynomial in $|L|$ thus in $\norm{L}_1$.
This means that they can be effectively solved when the given set of words is dense, that is, $\norm{L}_\infty$ is much smaller (e.g., logarithmic) than $|L|$.
We also derive upper bounds on the length of the same order.

We conclude that for a finite list $L$ of words over a fixed alphabet, $2^{\O(\norm{L}_\infty)}$ is a tight upper bound on both the length of the longest words that are not in $L^*$ when $L^*$ is cofinite and the length of the shortest incompletable words when $L^*$ is not factor universal.
Furthermore, in terms of $\norm{L}_1$, the subexponential length $2^{\varTheta(\sqrt[5]{\norm{L}_1})}$ is attainable.

\section{Preliminaries}\label{sec:preliminaries}

Let $\varepsilon$ denote the empty word.

A \emph{nondeterministic finite automaton} (\emph{NFA}) is a quintuple $\mathcal{A} = (Q_\mathcal{A},\Sigma,\delta_\mathcal{A},q_0,F_\mathcal{A})$, where $Q_\mathcal{A}$ is a finite non-empty set of \emph{states}, $\Sigma$ is a finite non-empty \emph{alphabet}, $\delta_\mathcal{A}\colon Q_\mathcal{A} \times (\Sigma \cup \{\varepsilon\}) \to 2^{Q_\mathcal{A}}$ is the \emph{transition function}, $q_0 \in Q_\mathcal{A}$ is the \emph{initial} state, and $F_\mathcal{A} \subseteq Q_\mathcal{A}$ is the set of \emph{final} states.
If for a state $q \in Q_\mathcal{A}$, the set of $\varepsilon$-transitions $\delta_\mathcal{A}(q,\varepsilon)$ is not explicitly defined, we assume $\delta_\mathcal{A}(q,\varepsilon)=\emptyset$ (no such transitions).

We extend $\delta_\mathcal{A}$ to a function $2^{Q_\mathcal{A}} \times \Sigma^* \to 2^{Q_\mathcal{A}}$ as usual.
We assume that the extended $\delta_\mathcal{A}$ is complemented by $\varepsilon$-transitions, i.e., for a subset $C \subseteq Q_\mathcal{A}$ and a word $w \in \Sigma^*$, $\delta_\mathcal{A}(C,w)$ is the set of all the states that can be obtained from a state in $C$ by applying sequentially a transition of each of the consecutive letters of $w$ interleaved with any number of $\varepsilon$-transitions, which can be used also at the beginning and the end.

A state $q' \in Q_\mathcal{A}$ is \emph{reachable from a state} $q \in Q_\mathcal{A}$ if there exists a word $w \in \Sigma^*$ such that $q' \in \delta_\mathcal{A}(\{q\},w)$.
Similarly, a subset $S' \subseteq Q_\mathcal{A}$ is \emph{reachable} from $S \subseteq Q_\mathcal{A}$ if there exists a word $w \in \Sigma^*$ such that $\delta_\mathcal{A}(S,w)=S'$.
Then we say that $q'$ (resp.\ $S'$) is \emph{reachable by the word} $w$ \emph{from} $q$ (resp.\ $S$).

An automaton \emph{accepts} a word $w \in \Sigma^*$ if $\delta_\mathcal{A}(\{q_0\},w) \cap F_\mathcal{A} \neq \emptyset$.
The set of all accepted words is the \emph{language} of the automaton.

A state $q \in Q_\mathcal{A}$ is called \emph{dead} if no final state is reachable from it, i.e., $\delta_\mathcal{A}(\{q\},w) \cap F_\mathcal{A} = \emptyset$ for all words $w \in \Sigma^*$.
Without affecting the language of an NFA, we can remove all dead states and remove all the transitions to them, i.e., replace $Q_\mathcal{A}$ with $Q_\mathcal{A} \setminus D$ and replace each $\delta_\mathcal{A}(q,a)$ with $\delta_\mathcal{A}(q,a) \setminus D$, where $q \in Q_\mathcal{A}$, $a \in \Sigma \cup \{\varepsilon\}$, and $D \subseteq Q_\mathcal{A}$ is the set of all dead states.

A special case of an NFA is a \emph{deterministic finite automaton} (\emph{DFA}), where for all $q \in Q_\mathcal{A}$ and $a \in \Sigma$ we have $|\delta_\mathcal{A}(q,a)|=1$ and there are no $\varepsilon$-transitions (i.e., $\delta_\mathcal{A}(q,\varepsilon)=\emptyset$).
In this case, we write $\delta_\mathcal{A}(q,a)=q'$ instead of $\delta_\mathcal{A}(q,a)=\{q'\}$.

\textbf{Automaton recognizing the Kleene star.}
We will use the well-known standard construction of an NFA recognizing the Kleene star of the language specified by a DFA (see, e.g., \cite{YuZhuangSalomaa1994}).
Let $\mathcal{A}=(Q_\mathcal{A},\Sigma,\delta_\mathcal{A},q_0,F_\mathcal{A})$ be a DFA.
Then $\mathcal{A}^*=(Q_{\mathcal{A}^*},\Sigma,\delta_{\mathcal{A}^*},q'_0,\{q'_0\})$ is the NFA obtained from $\mathcal{A}$ as follows.
The set of states $Q_{\mathcal{A}^*}$ is $Q_\mathcal{A} \cup \{q'_0\}$, where $q'_0$ is a fresh state.
The transition function $\delta_{\mathcal{A}^*}$ is defined as $\delta_\mathcal{A}(q,a)$ with some additional $\varepsilon$-transitions: we add a $\varepsilon$-transition from $q'_0$ to $q_0$ and from every final state in $F_\mathcal{A}$ to $q'_0$.
The language of the obtained NFA is $L^*$, where $L$ is the language of $\mathcal{A}$.

We further simplify the construction by removing all the dead states, thus our $\mathcal{A}^*$ do not contain them.
Also, if in $\mathcal{A}$ the initial state $q_0$ is not reachable from itself by any non-empty word (automata with this property are called \emph{non-returning} in the literature), then we can identify $q'_0$ with $q_0$, which is then the initial state and the unique final state.
We will use both simplifications in the constructions in this paper.

\section{Set rewriting system}\label{sec:set_rewriting}

We introduce \emph{set rewriting systems}, which are an auxiliary intermediate formalism for our further reductions.

\begin{definition}\label{def:set_rewriting_system}
A \emph{set rewriting system} is a pair $(P,R)$, where $P$ is a finite non-empty set of \emph{elements} and $R$ is a finite non-empty set of \emph{rules}.
A \emph{rule} is a function $r\colon P \to 2^P \cup \{\bot\}$.
\end{definition}

Given a set rewriting system and a subset $S \subseteq P$, a rule $r$ is \emph{legal} if $\bot \notin r(S)$ (i.e., there is no $s \in S$ such that $r(s)=\bot$).
The \emph{resulting subset} from applying a legal rule $r$ to $S$ is $S\cdot r = \bigcup_{s \in S} r(s)$.
Analogously, we inductively define that a sequence of rules $r_1,\ldots,r_k$ is \emph{legal} if $r_1,\ldots,r_{k-1}$ is legal for $S$ and $r_k$ is legal for $S \cdot r_1 \cdot \ldots \cdot r_{k-1}$.
The \emph{resulting subset} from applying a legal sequence of rules is $S \cdot r_1 \cdot \ldots \cdot r_k$.

Note that when a set rewriting system is given as the input for a problem, we can assume polynomial input size in terms of $|P|$ and $|R|$ (e.g., a straightforward encoding requires writing at most $|R|\cdot|P|^2$ elements from $P \cup \{\bot\}$).

\subsection{Immortality}

In general, immortality is a classical problem of whether there exists any configuration such that there is an infinite sequence of legally applied transitions to it.
This is in contrast to the usual setting, where the initial configuration is given and we ask about reachability.
In the case of systems with a bounded configuration space, this is equivalent to the existence of a cycle in the configuration space.
For instance, mortality (also called \emph{structural termination} in the literature) problems have been considered for Turing machines \cite{BlondelCassaigneNichitiu2002}, where the problem is undecidable, and for linearly bounded Turing machines with a counter \cite{BenAmram2015}, where the problem is PSPACE-complete.

Considering our setting, every set rewriting system contains a trivial cycle which is a loop on the empty set.
Therefore, in our mortality problem, we have to exclude the empty set as the cycle.
A set rewriting system $(P,R)$ is \emph{immortal} if there exists a non-empty subset $S \subseteq P$ and a non-empty sequence of rules $r_1,\ldots,r_k$ that is legal and yields $S$, i.e., $S\cdot r_1\cdot\ldots\cdot r_k = S$.
It is called \emph{mortal} otherwise.

Furthermore, we add the restriction that the empty set is not reachable from any non-empty subset; this will simplify reasoning in further reductions because otherwise reaching the empty set would need to be treated in a special way, as it does not count as a cycle but does not imply any restriction on any further rule applications.
A set rewriting system is \emph{non-emptiable} if for every non-empty subset $S \subseteq P$ and every rule $r \in R$, either $S\cdot r \neq \emptyset$ or $r$ is illegal for $S$.
This condition is equivalent to that for every element $p \in P$ and every rule $r \in R$, we have $r(p) \ne \emptyset$.
Hence, it is easy to check if a set rewriting system is non-emptiable.

\begin{problem}[Immortality of Set Rewriting]\label{pbm:immortality_set_rewriting}
Given a non-emptiable set rewriting system, is it immortal?
\end{problem}

First, we show that a mortal set rewriting system can admit exponentially long sequences of legal rules.

\begin{theorem}\label{thm:set_rewriting_maximal_legal_length}
For a mortal non-emptiable set rewriting system $(P,R)$, for every non-empty subset of $P$, the length of any legal sequence of rules is at most $2^{|P|}-2$.
For every $n \ge 1$, there exist a set rewriting system $(P,R)$ with $|P|=|R|=n$ and some subset of $P$ that meet the upper bound.
\end{theorem}
\begin{proof}
The upper bound is clear since there are $2^{|P|}-1$ distinct non-empty subsets and a legal sequence of length $2^{|P|}-2$ involves all of them.

To show tightness, we construct a set rewriting system $(P,R)$ with $n=|P|$ rules.
The elements will encode a specific binary counter.
Let $P = \{b_0,\ldots,b_{n-1}\}$.
For a subset $S \subseteq P$, we define $\mathit{val}(S,i) = 2^i$ if $b_i \in S$ and $\mathit{val}(S,i)=0$ otherwise, and we set the \emph{counter value} $\mathit{val}(S) = \sum_{0 \le i \le n-1}\ \mathit{val}(S,i)$.
For every $j \in \{0,\ldots,n-1\}$, we introduce a rule $r_j$ that, if it is legal, will increase the value of the counter by at least $1$.
The rules $r_j$ are defined as follows:
\begin{itemize}
\item $r_j(b_j) = \bot$;
\item $r_j(b_i) = \{b_j\}$ for $i \in \{0,1,\ldots,j-1\}$;
\item $r_j(b_i) = \{b_j,b_i\}$ for $i \in \{j+1,j+2,\ldots,n-1\}$.
\end{itemize} 

First, we observe that each legal rule $r_j$ applied to a non-empty set $S \subseteq P$ increases the counter value by at least $1$, i.e., $\mathit{val}(S) < \mathit{val}(S\cdot r_j)$.
It is because we know that $\mathit{val}(S,j) = 0$, as otherwise the rule would be illegal, and
\begin{alignat*}{4}
\mathit{val}(S\cdot r_j) & = \sum_{j < i < n}\ \mathit{val}(S\cdot r_j,i) + 2 ^ {j} && = \sum_{j < i < n}\ \mathit{val}(S,i) + 2^{j}\ && > \\
& > \sum_{j < i < n}\ \mathit{val}(S,i) + \sum_{0 \leq i < j}\ 2^{i}\ && \geq \sum_{0 \leq i < n}\ \mathit{val}(S,i) && =\ \mathit{val}(S).
\end{alignat*}

Second, we observe that for every non-empty $S \subsetneq P$, there exists a rule $r_j$ that increases the counter value exactly by $1$.
We choose the rule $r_j$ for $j$ being the smallest index such that $b_j \notin S$, and we have $\mathit{val}(S\cdot r_j) = \mathit{val}(S)+1$.
Furthermore, for $S = P$ there is no legal rule.

It follows that the set rewriting system is mortal, and for $S = \{b_0\}$, the longest possible legal sequence of rules has length $2^n-2$.
\end{proof}

Now, we show the PSPACE-completeness of the immortality problem.
The idea is a reduction from the non-universality of an NFA.
The NFA is combined with the counter developed for the proof of Theorem~\ref{thm:set_rewriting_maximal_legal_length}.
The NFA is encoded within the set rewriting system together with a counter incrementing its value with each transition.
The counter can be reset only if there exists a non-accepted word by the NFA, in which case it allows repeating a subset in the set rewriting system.

\begin{theorem}\label{thm:immortality_set_rewriting}
Problem~\ref{pbm:immortality_set_rewriting} (Immortality of Set Rewriting) is PSPACE-complete.
\end{theorem}
\begin{proof}
To solve the problem in NPSPACE thus in PSPACE, given a set rewriting system $(P,R)$, it is enough to guess a subset $S \subseteq P$ and a length $k \le 2^{|P|}$, and then to guess at most $k$ rules (storing only the current one), verifying whether the resulted subset is the same as $S$.

For PSPACE-hardness, we reduce from the non-universality problem for an NFA.
Given an NFA $\mathcal{A} = (Q_\mathcal{A},\Sigma,\delta_\mathcal{A},q_0,F_\mathcal{A})$, the question whether there exists any not accepted word over $\Sigma$ is PSPACE-complete (e.g., \cite[Section~10.6]{Aho&Hopcroft&Ullman:1974}).
We can assume that $\mathcal{A}$ does not have $\varepsilon$-transitions.

Let $n = |Q_\mathcal{A}|$.
We construct a set rewriting system $(P,R)$ of size polynomial in $n$.
As an ingredient, we use the counter from the proof of Theorem~\ref{thm:set_rewriting_maximal_legal_length}.
Let $P$ be the disjoint union of $Q_\mathcal{A}$ and $C = \{b_i \mid i \in \{0,1,\ldots,n-1\}\}$.
The elements of $C$ will encode the binary counter and for a subset $S \subseteq P$, we define $\mathrm{val}(S,i) = 2^i$ if $b_i \in S$ and $\mathrm{val}(S,i)=0$ otherwise, and we set $\mathrm{val}(S) = \sum_{0 \le i \le n-1}\ \mathrm{val}(S,i)$.

For every letter $a \in \Sigma$ and every $j \in \{0,1,\ldots,n-1\}$, we introduce a rule $r_{a,j}$ that acts as $a$ in the NFA on $Q_\mathcal{A}$ and, on the counter part, sets the $j$-th position of the counter.
The rules $r_{a,j}$ are defined as follows:
\begin{itemize}
\item $r_{a,j}(b_j) = \bot$;
\item $r_{a,j}(b_i) = \{b_j\}$ for $i \in \{0,1,2,\ldots,j-1\}$;
\item $r_{a,j}(b_i) = \{b_j,b_i\}$ for $i \in \{j+1,j+2,\ldots,n-1\}$;
\item $r_{a,j}(q) = \delta_\mathcal{A}(q,a) \cup \{b_j\}$ for $q \in Q_\mathcal{A}$.
\end{itemize} 
We also introduce the \emph{reset rule} that is defined as:
\begin{itemize}
\item $r_\mathrm{reset}(q) = 
\begin{cases}
    \bot,& \text{if}\ q \in F_\mathcal{A}; \\
    \{q_0,b_0\} & \text{otherwise}.
\end{cases}$
\end{itemize}

Now we observe the correctness.
Assume that there is a word not accepted by $\mathcal{A}$.
Note that if $w$ is a shortest non-accepted word, then $q_0 \notin \delta_\mathcal{A}(q_0,u)$ for all non-empty prefixes $u$ of $w$.
Hence, there exists a non-accepted word $w = a_1 a_2 \ldots a_m$ of length at most $2^{n-1}$.
As observed in the proof of Theorem~\ref{thm:set_rewriting_maximal_legal_length}, we know that for each value $x$ of the counter, there exists a rule that increments the counter value exactly by $1$.
Let $f(x)$ be the smallest index of a zero in the binary representation of $x$, where the least significant position of the binary representation is indexed by zero; hence a rule $r_{a_i,f(x)}$, if it is legal for $S$, increments the counter value of $S$ by $1$.
Then the set $S = \{b_0,q_0\} \cdot r_{a_1,f(1)} \cdot r_{a_2,f(2)} \cdot \ldots \cdot r_{a_m,f(m)}$ has the property that $\mathrm{val}(S) = m < 2^n$ and $S \cap F = \emptyset$, because $w$ is not accepted by $\mathcal{A}$.
Thus, rule $r_\mathrm{reset}$ is legal, so $\{b_0,q_0\} \cdot r_{a_1,f(1)} \cdot r_{a_2,f(2)} \cdot \ldots \cdot r_{a_m,f(m)} \cdot r_\mathrm{reset} = \{b_0,q_0\}$.
Hence the set rewriting system is immortal.

For the converse, assume that there exists a subset $S \subseteq P$ and a non-empty sequence of rules $r_{j_1}, r_{j_2}, \ldots, r_{j_m}$ such that $S \cdot r_{j_1} \cdot r_{j_2} \cdot \ldots \cdot r_{j_m} = S$.
As observed in the proof of Theorem~\ref{thm:set_rewriting_maximal_legal_length}, we know that every rule different from $r_\mathrm{reset}$ increments the counter value by at least $1$.
Hence, there must be some index $1 \le k \le m$ such that $r_{j_k} = r_\mathrm{reset}$.
Consider the sequence of rules $r_{j_1}, r_{j_2}, \ldots, r_{j_m}, r_{j_1}, r_{j_2} \ldots, r_{j_m}$.
In this sequence, $r_\mathrm{reset}$ appears at least twice.
Taking a shortest sequence of rules between any two $r_\mathrm{reset}$ rules (not including the reset rules), we get a sequence $r_{a_1,i_1}, r_{a_2,i_2}, \ldots, r_{a_d,i_d}$ such that $\{q_0, b_0\} \cdot r_{a_1,i_1} \cdot r_{a_2,i_2} \cdot \ldots \cdot r_{a_d,i_d} r_\mathrm{reset} = \{q_0, b_0\}$.
Since $r_\mathrm{reset}$ is legal when applied, the word $a_1 a_2 \ldots a_d$ is such that $\delta_\mathcal{A}(q_0,a_1 a_2 \ldots a_d) \cap F = \emptyset$ thus is not accepted by $\mathcal{A}$.
\end{proof}

\begin{lemma}\label{lem:set_rewriting_singleton}
If a rule $r$ is legal for a subset $S \subseteq P$, then it is also legal for every subset $S' \subseteq S$ and $S'\cdot r \subseteq S\cdot r$.
\end{lemma}

By this observation, when showing if the system is immortal, it is enough to consider only singleton subsets $S$ from which we start applying rules to find a cycle.
Although a singleton does not necessarily occur in a cycle, a non-emptiable set rewriting system is immortal if and only if, for some singleton, there exist arbitrary long legal sequences of rules.

\subsection{Emptying}\label{subsec:set_rewriting_emptying}

The second problem is the reachability of the empty set. 
This is related to factor universality and is necessary for our further reduction.

For a subset $S \subseteq P$, a sequence of rules $r_1,\ldots,r_k$ such that $S\cdot r_1\cdot \ldots \cdot r_k = \emptyset$ is called \emph{$S$-emptying}.

We call a set rewriting system \emph{permissive} if all rules are always legal.
In other words, all rules are legal for $P$.
A permissive set rewriting system $(P,R)$ is equivalent to the \emph{semi-NFA} whose set of states is $P$ and the alphabet is $R$; the NFA initial and final states are irrelevant.

\begin{problem}[Emptying Set Rewriting]\label{pbm:emptying_set_rewriting}
For a given permissive set rewriting system $(P,R)$, does there exist a $P$-emptying sequence of rules?
\end{problem}

Let $\mathcal{A} = (Q_\mathcal{A},\Sigma,\delta_\mathcal{A},q_0,F_\mathcal{A})$ be an NFA.
Analogously to set rewriting, for a subset $S \subseteq Q_\mathcal{A}$, a word $w \in \Sigma^*$ is called \emph{$S$-emptying} if $\delta_\mathcal{A}(S,w)=\emptyset$.

The following criterion for the factor universality of a language represented by an NFA is known.

\begin{proposition}[{\rm \cite{RampersadShallitXu2012}}]\label{pro:fu_criterion}
Let $\mathcal{A}=(Q_\mathcal{A},\Sigma,\delta_\mathcal{A},q_0,F_\mathcal{A})$ be an NFA such that every state is reachable from the initial state $q_0$ and there are no dead states.
Then a word $w$ is not a factor of a word accepted by $\mathcal{A}$ if and only if $w$ is $Q_\mathcal{A}$-emptying.
The language of $\mathcal{A}$ is factor universal if and only if there does not exist a $Q_\mathcal{A}$-emptying word.
\end{proposition}

It is known that the problem of whether a given language specified by an NFA is factor universal is PSPACE-complete \cite{RampersadShallitXu2012}.
Since it is also easy to solve Problem~\ref{pbm:emptying_set_rewriting} in PSPACE, we have:

\begin{proposition}\label{pro:emptying_PSPACE}
Problem~\ref{pbm:emptying_set_rewriting} (Emptying Set Rewriting) is PSPACE-complete.
\end{proposition}

Additionally, we will need an exponential lower bound on the length of the shortest $P$-emptying sequences of rules.
For this, we also develop a specific counter, but now counting downwards and allowing to decrease the value by at most $1$.
Instead of rules being illegal, undesired rule applications reset the counter to the maximal value.

\begin{theorem}\label{thm:long_emptying_set_rewriting}
For a permissive set rewriting system $(P,R)$, if there exists a $P$-emptying sequence of rules, then the shortest such sequences have length at most $2^{|P|}-1$.
For every $n \ge 1$, there exists a set rewriting system $(P,R)$ with $|P|=|R|=n$ that meets the bound.
\end{theorem}
\begin{proof}
The upper bound $2^{|P|}-1$ is trivial.

For every $n \ge 1$, we construct a permissive set rewriting system $(P,R)$, which represents a binary counter of length $n$.
Let $P = \{b_i \mid i \in \{0,1,\ldots,n-1\}\}$.
For a subset $S \subseteq P$, we define $\mathrm{val}(S,i) = 2^i$ if $b_i \in S$ and $\mathrm{val}(S,i) = 0$ otherwise, and $\mathrm{val}(S) = \sum_{0 \le i \le n-1}\ \mathrm{val}(S,i)$.

We define the rules that allow decreasing the value of the counter by $1$.
If a wrong rule is used, the counter is reset to its maximal value.
The set of rules $R$ consists of rules $r_j$ for $j \in \{0,1,\ldots,n-1\}$, where each $r_j$ is defined as follows:
\begin{enumerate}
\item $r_j(b_i) = P$ for $i \in \{0,1,\ldots,j-1\}$;
\item $r_j(b_j) = \{b_i \mid i \in \{0,1,\ldots,j-1\}\}$;
\item $r_j(b_i) = \{b_i\}$ for $i \in \{j+1,j+2,\ldots,n-1\}$.
\end{enumerate}

We observe that emptying this set rewriting system corresponds to setting the counter value to $0$.
For a subset $S$, let $i$ be the smallest index such that $b_i \in S$.
Then for all the smaller positions $j < i$, $b_j \notin S$.
Notice that for all rules $r_k$ for $k \in \{1,2,\ldots,n-1\} \setminus \{i\}$, we have $\mathrm{val}(S\cdot r_k) \geq \mathrm{val}(S)$.
This is because if $k < i$, then $S \cdot r_k = S$, and if $k > i$, then $S \cdot r_k = P$.
Hence, the only rule that decreases the counter is $r_i$, and then $\mathrm{val}(S\cdot r_i) = \mathrm{val}(S) - 1$.
Thus, the shortest sequence of rules that is $P$-emptying has length $2^n-1$.
\end{proof}

\section{The Frobenius monoid problem}\label{sec:fpfl}

We note the known result about the PSPACE-membership.

\begin{proposition}[{\rm\cite[Corollary~5.5.8]{Xu2009FPFM}}]
Problem~\ref{pbm:fpfl} is in PSPACE.
\end{proposition}

For PSPACE-hardness, we reduce from Problem~\ref{pbm:immortality_set_rewriting} (Immortality of Set Rewriting) to Problem~\ref{pbm:fpfl} (Frobenius Monoid Problem for a Finite Set of Words).
In the first step, we reduce to the case where $L$ is specified as a DFA instead of a list of words.
Then we binarize the DFA, and finally, we count the number of words in the language to bound the size of the list of words.

\subsection{The DFA construction}\label{subsec:fpfl-dfa}

As the input for the reduction, we take a non-emptiable set rewriting system $(P,R)$.
Without loss of generality, we assume the set of elements $P = \{p_1,p_2,\ldots,p_\ell\}$ and the rules $R = \{r_1,r_2,\ldots,r_m\}$.

We construct a DFA $\mathcal{A} = (Q_\mathcal{A},\Sigma,\delta_\mathcal{A},q_0,F)$ recognizing a finite language $L$ such that $L^*$ is not cofinite if and only if $(P,R)$ is immortal.
The number and the lengths of words in $L$ will be polynomial, which will allow further polynomial reduction to the case of a list of words.
First, we define the DFA; then we describe its mechanism, and finally, we prove the correctness formally.

\begin{figure}[htb]\begin{center}
\resizebox{12cm}{!}{\large
\tikzset{every state/.style={minimum size=1.35cm}}
\tikzset{every path/.style={-{Latex[width=1mm,length=1.75mm]}}}

\tikzstyle{accepting}=[path picture={
\draw let 
  \p1 = (path picture bounding box.east),
  \p2 = (path picture bounding box.center)
  in
    (\p2) circle (\x1 - \x2 - 2pt);
 }]
\begin{tikzpicture}
\node[state, inner sep=0pt] (main_0) {$q_0$};
\node[draw=none] (tmp_draw_arrow) [left=0.4cm of main_0] {};
\path[] (tmp_draw_arrow) edge (main_0);
\node[state, inner sep=0pt] (main_1) [right= of main_0]  {$p_1$};
\node[state, inner sep=0pt] (main_2) [right= of main_1]  {$p_2$};
\node[state, inner sep=0pt, draw=none] (main_3) [right= of main_2]  {$\cdots$};
\node[state, inner sep=0pt] (main_4) [right= of main_3]  {$p_\ell$};
\path[]
(main_0) edge node [above] {$\alpha$} (main_1)
(main_1) edge node [above] {$\alpha$} (main_2)
(main_2) edge node [above] {$\alpha$} (main_3)
(main_3) edge node [above] {$\alpha$} (main_4)
;
\foreach \i in {1,2}
{
\node[state, inner sep=0pt,fill=white] (s^{\i,4}_l) [below right=1.05 and -0.6cm of main_\i] {$s^{\i,m}_{\ell}$};
\node[state, inner sep=0pt,draw=none] (s^{\i,4}_l-1) [below= of s^{\i,4}_l] {$\vdots$};
\node[state, inner sep=0pt,fill=white] (s^{\i,4}_1) [below= of s^{\i,4}_l-1] {$s^{\i,m}_1$};
\path[]
(main_\i) edge node [ below left=-0.1cm and 0.23cm] {\scalebox{0.75}{$r_m$}} (s^{\i,4}_l)
(s^{\i,4}_l) edge node [below left=0.15cm] {} (s^{\i,4}_l-1)
;
\node[state, inner sep=0pt, fill=white] (s^{\i,3}_l) [below right=1.3 and -0.85cm of main_\i] {\scalebox{1.0}{$\Ddots$}};
\node[state, inner sep=0pt,draw=none] (s^{\i,3}_l-1) [below= of s^{\i,3}_l] {$\vdots$};
\node[state, inner sep=0pt, draw=none] (s^{\i,3}_1) [below= of s^{\i,3}_l-1] {};
\path[]
(main_\i) edge node [ below left=-0.10cm and 0.4cm] {\scalebox{0.6}{$\Ddots$}} (s^{\i,3}_l)
(s^{\i,3}_l) edge node [below left=0.15cm] {} (s^{\i,3}_l-1)
;
\node[state, inner sep=0pt,fill=white] (s^{\i,2}_l) [below right=1.55 and -1.10cm of main_\i] {$s^{\i,2}_{\ell}$};
\node[state, inner sep=0pt,draw=none] (s^{\i,2}_l-1) [below= of s^{\i,2}_l] {$\vdots$};
\node[state, inner sep=0pt, draw=none] (s^{\i,2}_1) [below= of s^{\i,2}_l-1] {};
\path[]
(main_\i) edge node [ below left=0cm and 0.4cm] {\scalebox{0.75}{$r_2$}} (s^{\i,2}_l)
(s^{\i,2}_l) edge node [below left=0.15cm] {} (s^{\i,2}_l-1)
;
\node[state, inner sep=0pt,fill=white] (s^{\i,1}_l) [below right=1.80 and -1.35cm of main_\i] {$s^{\i,1}_{\ell}$};
\node[state, inner sep=0pt,draw=none] (s^{\i,1}_l-1) [below= of s^{\i,1}_l] {$\vdots$};
\node[state, inner sep=0pt,draw=none] (s^{\i,1}_1) [below= of s^{\i,1}_l-1] {};
\path[]
(main_\i) edge node [below left=0cm and 0.4cm] {\scalebox{0.75}{$r_1$}} (s^{\i,1}_l)
(s^{\i,1}_l) edge node [ left=0.15cm] {$\alpha$} (s^{\i,1}_l-1)
;
}
\foreach \i in {4}
{
\node[state, inner sep=0pt,fill=white] (s^{\i,4}_l) [below right=1.05 and -0.6cm of main_\i] {$s^{\ell,m}_{\ell}$};
\node[state, inner sep=0pt,draw=none] (s^{\i,4}_l-1) [below= of s^{\i,4}_l] {$\vdots$};
\node[state, inner sep=0pt,fill=white] (s^{\i,4}_1) [below= of s^{\i,4}_l-1] {$s^{\ell,m}_1$};
\path[]
(main_\i) edge node [ below left=-0.1cm and 0.23cm] {\scalebox{0.75}{$r_m$}} (s^{\i,4}_l)
(s^{\i,4}_l) edge node [below left=0.15cm] {} (s^{\i,4}_l-1)
;
\node[state, inner sep=0pt,fill=white] (s^{\i,3}_l) [below right=1.3 and -0.85cm of main_\i] {};
\node[state, inner sep=0pt,draw=none] (s^{\i,3}_l-1) [below= of s^{\i,3}_l] {$\vdots$};
\node[state, inner sep=0pt, draw=none] (s^{\i,3}_1) [below= of s^{\i,3}_l-1] {};
\path[]
(main_\i) edge node [ below left=-0.10cm and 0.4cm] {\scalebox{0.6}{$\Ddots$}} (s^{\i,3}_l)
(s^{\i,3}_l) edge node [below left=0.15cm] {} (s^{\i,3}_l-1)
;
\node[state, inner sep=0pt,fill=white] (s^{\i,2}_l) [below right=1.55 and -1.10cm of main_\i] {$s^{\ell,2}_{\ell}$};
\node[state, inner sep=0pt,draw=none] (s^{\i,2}_l-1) [below= of s^{\i,2}_l] {$\vdots$};
\node[state, inner sep=0pt, draw=none] (s^{\i,2}_1) [below= of s^{\i,2}_l-1] {};
\path[]
(main_\i) edge node [ below left=0cm and 0.4cm] {\scalebox{0.75}{$r_2$}} (s^{\i,2}_l)
(s^{\i,2}_l) edge node [below left=0.15cm] {} (s^{\i,2}_l-1)
;
\node[state, inner sep=0pt,fill=white] (s^{\i,1}_l) [below right=1.80 and -1.35cm of main_\i] {$s^{\ell,1}_{\ell}$};
\node[state, inner sep=0pt,draw=none] (s^{\i,1}_l-1) [below= of s^{\i,1}_l] {$\vdots$};
\node[state, inner sep=0pt,draw=none] (s^{\i,1}_1) [below= of s^{\i,1}_l-1] {};
\path[]
(main_\i) edge node [below left=0cm and 0.4cm] {\scalebox{0.75}{$r_1$}} (s^{\i,1}_l)
(s^{\i,1}_l) edge node [ left=0.15cm] {$\alpha$} (s^{\i,1}_l-1)
;
}
\node[state, inner sep=0pt,draw=none] (s^{3,4}_l) [below right=1.05 and -0.6cm of main_3] {$\cdots$};
\node[state, inner sep=0pt,draw=none] (s^{3,4}_l-1) [below= of s^{3,4}_l] {$\ddots$};
\node[state, inner sep=0pt,draw=none] (s^{3,4}_1) [below= of s^{3,4}_l-1] {$\cdots$};
\node[state, inner sep=0pt,draw=none] (s^{3,3}_l) [below right=1.3 and -0.85cm of main_3] {$\cdots$};
\node[state, inner sep=0pt,draw=none] (s^{3,3}_l-1) [below= of s^{3,3}_l] {$\ddots$};
\node[state, inner sep=0pt,draw=none] (s^{3,3}_1) [below= of s^{3,3}_l-1] {$\cdots$};
\node[state, inner sep=0pt,draw=none] (s^{3,2}_l) [below right=1.55 and -1.10cm of main_3] {$\cdots$};
\node[state, inner sep=0pt,draw=none] (s^{3,2}_l-1) [below= of s^{3,2}_l] {$\ddots$};
\node[state, inner sep=0pt,draw=none] (s^{3,2}_1) [below= of s^{3,2}_l-1] {$\cdots$};
\node[state, inner sep=0pt,draw=none] (s^{3,0}_l) [below right=1.80 and -1.35cm of main_3] {$\cdots$};
\node[state, inner sep=0pt,draw=none] (s^{3,0}_l-1) [below= of s^{3,0}_l] {$\ddots$};
\node[state, inner sep=0pt,draw=none] (s^{3,0}_1) [below= of s^{3,0}_l-1] {$\cdots$};
\node[state, inner sep=0pt] (g) [below right=1.5cm and 0.5cm of s^{2,1}_1] {$q_\mathrm{g}$};
\node[state, inner sep=0pt] (q_s) [right =2.3cm of g] {$q_\mathrm{s}$};
\path[]
(g) edge node [above] {$R$} (q_s)
;
\foreach \i in {1,2}
{
\foreach \j in {1,3,4}
{
\path[]
(s^{\i,\j}_1) edge node [below] {} (g)
;
}
}
\path[]
(s^{1,2}_1) edge node [below left=0.08cm] {$\alpha$} (g)
;
\path[]
(s^{2,2}_1) edge node [right=0.2925cm] {$\alpha$} (g)
;
\path[]
(s^{4,4}_1)+(-2.75mm,0) edge node [below] {} ($(g)+(0.49cm,0.49cm)$)
;
\path[]
(s^{4,3}_1)+(-2mm,0) edge node [below] {} ($(g)+(0.528cm,0.447cm)$)
;
\path[]
(s^{4,2}_1)+(-1mm,0) edge node [below] {} ($(g)+(0.565cm,0.400cm)$)
;
\path[]
(s^{4,1}_1)+(0.25mm,0) edge node [below=0.08cm] {$\alpha$} ($(g)+(0.597cm,0.350cm)$) 
;
\foreach \j in {3,2,1}
{
\foreach \i in {1,2}
{
\node[state, inner sep=0pt,fill=white] (ts^{\i,\j}_1) [below= of s^{\i,\j}_l-1] {$s^{\i,\j}_1$};
}
\node[state, inner sep=0pt,fill=white] (ts^{4,\j}_1) [below= of s^{4,\j}_l-1] {$s^{\ell,\j}_1$};
}
\path[] (q_s) edge [out=20,in=-20,distance=1cm] node [below=.65cm,left] {$\Sigma$} (q_s);
\node[state, accepting,  inner sep=0pt] (f_0) [below =11.0cm of main_0] {$f_0$};
\node[state, accepting,  inner sep=0pt] (f_1) [right=of f_0] {$f_1$};
\node[state , draw=none] (f_2) [right=of f_1] {$\cdots$};
\node[state, accepting,  inner sep=0pt,] (f_3) [right=of f_2]  {$f_{\ell-1}$};
\node[state, accepting,  inner sep=0pt] (f_4) [right=of f_3]  {$f_\ell$};
\path[]
(f_0) edge node [above] {$\alpha$} (f_1)
(f_1) edge node [above] {$\alpha$} (f_2)
(f_2) edge node [above] {$\alpha$} (f_3)
(f_3) edge node [above] {$\alpha$} (f_4)
;
\foreach \i in {1,2,4}
{
\foreach \j in {4,3,2}
{
\path[]
(s^{\i,\j}_l-1) edge node [ left=0.15cm] {} (s^{\i,\j}_1);
}
\path[]
(s^{\i,1}_l-1) edge node [ left=0.15cm] { $\alpha$ } (s^{\i,1}_1);
}
\foreach \i in {0}
{
\path[]
(f_\i) edge [out=50, in=205] node [above] {$R$} (q_s)
;
}
\path[]
(f_1) edge [out=50, in=220] node [below] {$R$} (q_s)
;
\path[]
(f_3) edge [out=50, in=250] node [above = 0.001cm] {$R$} (q_s)
;
\path[]
(f_4) edge node [right] {$\Sigma$} (q_s)
;
\path[] (main_0) edge node [left] {$R$} (f_0);
\path[] (g) edge [out=180, in=80] node [above left] {$\alpha$} (f_0);

\node[draw=none,coordinate] (main_4_right) [right=of main_4] {};
\node[draw=none,coordinate] (main_4_right_below) [below=of main_4_right] {};
\node[draw=none,coordinate] (f_4_right) [right=of f_4] {};
\node[draw=none,coordinate] (f_4_below) [below=of f_4] {};
\node[draw=none,coordinate] (f_4_right_below) [below=of f_4_right] {};
\node[draw=none,coordinate] (f_0_below) [below=of f_0] {};
\node[draw=none,coordinate] (f_0_below_right) [right=of f_0_below] {};
\path[every edge] (main_4.east) to[out=0,in=90] (main_4_right_below) -- (f_4_right) to[out=270,in=0] (f_4_below) -- (f_0_below_right) to[out=180,in=270] (f_0.south);
\node[draw=none] (arrow_label) [below right=0.3cm and 0.3cm of f_4] {$\alpha$};
\end{tikzpicture}}
\end{center}\caption{The scheme of the DFA $\mathcal{A}$ for a set rewriting system. The transitions from all $s^{i,j}_h$ to $f_0$ on $R$ are not drawn.}\label{fig:fpfl_DFA}
\end{figure}

The DFA is presented in Fig.~\ref{fig:fpfl_DFA}.
Since this is a DFA construction, for every letter and every state there should be a transition; for a clearer picture, we have omitted drawing the transitions from the setting states to $f_0$ on $R$.

The alphabet of $\mathcal{A}$ is $\Sigma = R \cup \{\alpha\}$, where the letters from $R$ are \emph{rule letters}, and $\alpha$ is a fresh special letter that will be used to shift the states and separate applications of rule letters.
The set of states $Q_\mathcal{A}$ is the disjoint sum of the following sets:
\begin{itemize}
\item $\{q_0\}$; the initial state.
\item $Q_\mathrm{P} = P$; the elements of the set rewriting system.
\item $Q_\mathrm{F} = \{f_i \mid i \in \{0,1,\ldots,\ell\}\}$; the \emph{forcing states}.
\item $\{s^{i,j}_h \mid i,h \in \{1,2,\ldots,\ell\} \land j \in \{1,2,\ldots,m\} \land r_j(p_i) \neq \bot\}$; the \emph{setting states}; a setting state $s^{i,j}_h$ is dedicated to the element $p_i$ and the rule $r_j$.
\item $\{q_\mathrm{g}\}$; the \emph{guard state}.
\item $\{q_\mathrm{s}\}$; the \emph{sink state}, which is the unique dead state.
\end{itemize}

The transition function $\delta_\mathcal{A}$ is defined as follows (see also Fig.~\ref{fig:fpfl_DFA}):
\begin{itemize}
\item $\delta_\mathcal{A}(q_0,\alpha) = p_1$.
\item $\delta_\mathcal{A}(p_i,\alpha) = p_{i+1}$ for all $i \in \{0,1,\ldots,\ell-1\}$.
\item $\delta_\mathcal{A}(p_\ell,\alpha) = f_0$.
\item $\delta_\mathcal{A}(p_i,r_j) = \begin{cases}
  s^{i,j}_\ell, & \text{if}\ r_j(p_i) \neq \bot \\
  f_0, & \text{otherwise}
\end{cases}$ \\
  for all $i \in \{1,2,\ldots,\ell\}$ and $j \in \{1,2,\ldots,m\}$; 
  the transition of a rule letter maps the elements from $Q_\mathrm{P}$ either to the first setting state in the dedicated chain, when the rule is legal, or directly to $f_0$, otherwise.
\item $\delta_\mathcal{A}(q_0,r_j) = f_0$ for all $j \in \{1,2,\ldots,m\}$.
\item $\delta_\mathcal{A}(s^{i,j}_h,\alpha) = s^{i,j}_{h-1}$ for all $i \in \{1,2,\ldots,\ell\}$, $j \in \{1,2,\ldots,m\}$, and $h \in \{\ell,\ell-1,\ldots,2\}$.
\item $\delta_\mathcal{A}(s^{i,j}_1,\alpha) = q_\mathrm{g}$ for all $i \in \{1,2,\ldots,\ell\}$ and $j \in \{1,2,\ldots,m\}$.
\item $\delta_\mathcal{A}(s^{i,j}_h, r_k) = f_0$ for all $i,h \in \{1,2,\ldots,\ell\}$ and $j,k \in \{1,2,\ldots,m\}$.
\item $\delta_\mathcal{A}(q_\mathrm{g},\alpha) = f_0$.
\item $\delta_\mathcal{A}(q_\mathrm{g},r_j) = q_\mathrm{s}$ for all $j \in \{1,2,\ldots,m\}$.
\item $\delta_\mathcal{A}(f_i,\alpha) = f_{i+1}$ for all $i \in \{0,1,\ldots,\ell-1\}$.
\item $\delta_\mathcal{A}(f_i,r_j) = q_\mathrm{s}$ for all $i \in \{0,1,\ldots,\ell\}$ and $j \in \{1,2,\ldots,\ell\}$.
\item $\delta_\mathcal{A}(f_\ell,\alpha) = q_\mathrm{s}$.
\end{itemize}

The set of final states $F_\mathcal{A}$ is the disjoint union of the following sets:
\begin{itemize}
\item $Q_\mathrm{F}$; all forcing states are final.
\item $\{s^{i,j}_h \mid i,h \in \{1,2,\ldots,\ell\} \land j \in \{1,2,\ldots,m\} \land r_j(p_i) \neq \bot \land p_h \in r_j(p_i)\}$;
states in a setting chain are final according to the rule $r_j$ applied to the element $p_i$.
\end{itemize}

\subsubsection{The mechanism}

We describe the idea of the construction and introduce a few necessary notions for the analysis.

To reason about $L^*$, we construct the NFA $\mathcal{A}^*=(Q_{\mathcal{A}^*},\Sigma,\delta_{\mathcal{A}^*},\{q_0\},\{q_0\})$ recognizing the language $L^*$ as described in Section~\ref{sec:preliminaries}.
This NFA is $\mathcal{A}$ with added $\varepsilon$-transitions from every final state to $q_0$ and with the dead state $q_\mathrm{s}$ removed.

Given a word $w$ and a subset $C$ by the context, the states that are obtained by applying the action of $w$ to $C$, i.e., $\delta_{\mathcal{A}^*}(C,w)$, are called \emph{active}.

A word $w \in \Sigma^*$ is \emph{irrevocably accepted} if for every $u \in \Sigma^*$, the word $wu$ belongs to $L^*$.
One of the key properties of our construction is that all words such that the state $f_0$ becomes active when starting from $\{q_0\}$ are irrevocably accepted.
This means that for a non-accepted word $w$, state $f_0$ (as every other forcing state) cannot be activated by the application of any prefix of $w$ to the initial subset $\{q_0\}$.

In general, $\mathcal{A}^*$ simulates the set rewriting system.
A subset of $Q_\mathrm{P}$ corresponds to the same subset of elements in the set rewriting system.
Sequences of rules translate to words of a specific form.
When a sequence of rules is not legal for a subset, the corresponding word activates $f_0$ at some point.
The same holds for a word that violates the specific form.
This provides a correspondence between sequences of legal rules and possibly non-accepted words.

For a subset $S \subseteq Q_\mathrm{P}$, applying the word $r_j \alpha^\ell$, for some $r_j \in R$, corresponds to applying the rule $r_j$ in the set rewriting system for $S$, i.e., $S' = \delta_{\mathcal{A}^*}(S,r_j \alpha^\ell)=(S\cdot r_j) \cup \{q_\mathrm{g}\}$ when $r_j$ is legal for $S$.
The chain of the setting states for a state $p_i \in Q_\mathrm{P}$ has its final states in the DFA defined accordingly to the action of the rule $r_j$ for the element $p_i$.
The indices keep the correspondence that if a state $s^{i,j}_h$ is final in the DFA, then after applying $r_j \alpha^\ell$, the state $p_h$ becomes active if and only if $p_i \in S$ (assuming that the rule $r_j$ is legal for $S$).
In the end, the guard state $q_\mathrm{g}$ is additionally activated, which ensures that we must use a rule letter as the next one since the transition of $\alpha$ maps the guard state to $f_0$.
Without the guard state, sometimes one could use more $\alpha$ letters to shift the states within $Q_\mathrm{P}$ and in this way cheat by obtaining a different subset of $Q_\mathrm{P}$.
If the rule $r_j$ is not legal for $S$, then the transitions of $r_j$ directly activate $f_0$.

A word $u$ is \emph{simulating for a subset} $S \subseteq Q_\mathrm{P}$ if it is in the form of $r_{i_1} \alpha^\ell r_{i_2} \alpha^\ell \cdot \ldots \cdot r_{i_k} \alpha^\ell$ (for $k \ge 0$) and the sequence of the rules $r_{i_1}, r_{i_2}, \ldots, r_{i_k}$ is legal for $S$ in the set rewriting system.
The actions of these words correspond to applying the contained sequence of rules.
The construction ensures that, for a subset containing a non-empty $S \subseteq Q_\mathrm{P}$ together with the guard state $q_\mathrm{g}$, using a simulating word is the only possibility to avoid $f_0$.

A special case occurs at the beginning, i.e., for the initial subset $\{q_0\}$.
To avoid activation of $f_0$, we must start with $\alpha^i$ for any $1 \le i \le \ell$.
Then we obtain the subset $\delta_{\mathcal{A}^*}(\{q_0\},\alpha^i)=\{p_i\}$; this corresponds to the selection of the initial singleton in the set rewriting system (cf.\ Lemma~\ref{lem:set_rewriting_singleton}).
After that, a simulating word must be used, unless $f_0$ is activated.
Note that $Q_\mathrm{P} \cup \{q_\mathrm{g}\}$ does not contain final states, thus in this way we get non-accepted words.
It follows that we can find arbitrarily long such words if and only if the set rewriting system is immortal.

\subsubsection{Correctness}

We prove the correctness formally through the following lemmas.

The first lemma states that whenever $f_0$ becomes active, all subsequent words will be accepted, thus $f_0$ must be avoided when constructing a non-accepted word.

A word $w \in \Sigma^*$ is \emph{$f_0$-omitting for a subset} $C \subseteq Q_{\mathcal{A}^*}$ if there is no prefix $u$ of $w$ such that $f_0 \in \delta_{\mathcal{A}^*}(C,u)$.
It is simply \emph{$f_0$-omitting} if it is $f_0$-omitting for $\{q_0\}$.

\begin{lemma}\label{lem:fpfl-irrevocably_accepted}
If a word $w \in \Sigma^*$ is not $f_0$-omitting, then it is irrevocably accepted.
\end{lemma}
\begin{proof}
If $w \in \Sigma^*$ is not $f_0$-omitting, then there is a prefix $u$ of $w$ such that $f_0 \in \delta_{\mathcal{A}^*}(\{q_0\},u)$.
It is enough to observe that for every word $v \in \Sigma^*$, the set $\delta_{\mathcal{A}^*}(\{f_0\},v)$ contains a forcing state.
All forcing states are final, thus $uv$ and, in particular, all words containing $w$ as a prefix will be accepted.
Suppose this is not the case, and let $v$ be a shortest word such that $\delta_{\mathcal{A}^*}(\{f_0\},v)$ does not contain any forcing states.
Then for every non-empty proper prefix $v'$ of $v$, $\delta_{\mathcal{A}^*}(\{f_0\},v')$ does not contain $f_0$, as otherwise the suffix $v''$, where $v'v''=v$, would be a shorter word than $v$ with the same property.
Note that $\delta_{\mathcal{A}^*}(\{f_0\},\varepsilon)=\{f_0,q_0\}$, and in general, whenever a forcing state is active, $q_0$ is also active.
Thus the only possibility for $v$ is to start with $\alpha^{\ell+1}$; otherwise, active state $q_0$ would be mapped to $f_0$ by the transition of a rule letter after $\alpha^i$ for an $i \le \ell$.
However, the action of $\alpha^{\ell+1}$ through the chain on $Q_\mathrm{P}$ also maps $q_0$ to $f_0$, which yields a contradiction.
\end{proof}

Applying a simulating word corresponds to applying the sequence of rules that is contained in it.
The following lemma formalizes this claim.

\begin{lemma}\label{lem:fpfl-simulating_word}
Let $C \subseteq Q_\mathrm{P} \cup \{q_\mathrm{g}\}$ be such that $S = C \cap Q_\mathrm{P}$ is non-empty, and let $w = r_{i_1} \alpha^\ell \ldots r_{i_k} \alpha^\ell$ be a simulating word for $S$.
Then $\delta_{\mathcal{A}^*}(C,w) = (S \cdot r_{i_1}\cdot\ldots\cdot r_{i_k}) \cup \{q_\mathrm{g}\}$.
\end{lemma}
\begin{proof}
Let $C$ and $S$ be as in the lemma, and let $r_j$ be a rule that is legal for $S$.
The transitions of the letter $r_j$ map each state $p_i \in S$ to $s^{i,j}_\ell$.
Then the action of $\alpha^\ell$ maps these active states along the setting chains, activating state $q_0$ whenever an active setting state is final.
Eventually, they are mapped to $q_\mathrm{g}$.
A state $s^{i,j}_h$ is final if and only if $p_h \in r_j(p_i)$.
From the construction, if $s^{i,j}_h$ is final, then $q_0$ becomes active after applying $\alpha^{\ell-h}$.
Then $q_0$ is mapped to $p_h$ by the action of the remaining $\alpha^h$.
After the last occurrence of $\alpha$, the last active setting states are mapped to the guard state $q_\mathrm{g}$, and there is at least one such active state since $S$ is non-empty.
Finally, if besides $S$, $C$ contains the guard state, the action of $r_j$ deactivates it (maps to the empty set).
Hence, we have $\delta_{\mathcal{A}^*}(C,r_j \alpha^\ell) = (S \cdot r_j) \cup \{q_\mathrm{g}\}$.

Since the set rewriting system is non-emptiable, the set $S \cdot r_j$ is non-empty, thus we can apply the argument iteratively.
Hence, the lemma follows by induction on $k$.
\end{proof}

We show that, unless $f_0$ is activated, a word applied to a subset $C \subseteq Q_\mathrm{P} \cup \{q_\mathrm{g}\}$ must be a prefix of a simulating word for $S = C \cap Q_\mathrm{P}$.
The required condition is that the guard state is also present in $C$, so one cannot shift the states on $Q_\mathrm{P}$ by using $\alpha$.

\begin{lemma}\label{lem:fpfl-must_be_simulating}
Let $C = S \cup \{q_\mathrm{g}\}$, where $S \subseteq Q_\mathrm{P}$ is non-empty.
If $w$ is $f_0$-omitting for $C$, then $w$ is a prefix of a simulating word for $S$.
\end{lemma}
\begin{proof}
First, we observe that every word $w$ which does not activate $f_0$ starting from $C$, unless it is the empty word, must start with a rule letter $r_j$, since using $\alpha$ maps $q_\mathrm{g}$ to $f_0$ and we have assumed $q_\mathrm{g} \in C$.
Additionally, $r_j$ must be legal for $S$, as otherwise $f_0$ would be activated.
Afterwards, at least one of the first setting states must be active, because $S \neq \emptyset$.
Hence $\alpha^\ell$ must be used, unless $w$ ends before that and thus is a prefix of this pattern.
By Lemma~\ref{lem:fpfl-simulating_word} for $C$ and $w = r_j \alpha^\ell$, we know that the set of active states is now $(S\cdot r_j) \cup \{q_\mathrm{g}\}$.
By iterating this argument, we conclude that between each rule letter there must be exactly $\ell$ letters $\alpha$, it must start with a rule letter, and at the end, there are at most $\ell$ letters $\alpha$.
Furthermore, the rule letters must form a legal sequence of rules for $S$.
Therefore, we know that word $w$ has to be a prefix of some simulating word for $S$.
\end{proof}

In the beginning, before we may apply a simulating word, we can choose an arbitrary singleton $\{p_i\}$ as the initial subset.
Then a simulating word must be applied, as otherwise $f_0$ is activated.

\begin{lemma}\label{lem:fpfl-initial_must_be_simulating}
If a word $w$ is $f_0$-omitting, then $w$ is a prefix of $\alpha^i w'$, where $1 \le i \le \ell$ and $w'$ is a simulating word for $\{p_i\}$.
\end{lemma}
\begin{proof}
Let $w$ be a $f_0$-omitting word, and write $w = \alpha^i w'$ for $i \ge 0$ and $w' \in \Sigma^*$ that does not start with $\alpha$.
Since we start from $\{q_0\}$, we know that $1 \le i \le \ell$ unless $w$ is empty.
We have $\delta_{\mathcal{A}^*}(\{q_0\},\alpha^i) = \{p_i\}$.
Then $w'$ begins with some rule letter $r_j$, which must be a legal rule for $\{p_i\}$ in the set rewriting system, followed by $\alpha^\ell$, unless $w'$ is shorter and thus is a prefix of this pattern.

Hence $w = \alpha^i r_j \alpha^\ell u$.
By Lemma~\ref{lem:fpfl-simulating_word}, we have $C = \delta_{\mathcal{A}^*}(\{q_0\},\alpha^i r_j \alpha^\ell) = S \cup \{q_\mathrm{g}\}$ for $S = \{p_i\}\cdot r_j$.
Since the set rewriting is non-emptiable, $S \neq \emptyset$.
By Lemma~\ref{lem:fpfl-must_be_simulating} applied to $C$, we know that $u$ must be a prefix of a simulating word for $S$.
It follows that $r_j \alpha^\ell u$ is a prefix of a simulating word for $\{p_i\}$.
\end{proof}

Finally, we show the equivalence between the immortality of the set rewriting system and the non-cofiniteness of the language of $\mathcal{A}^*$.

\begin{lemma}\label{lem:fpfl-equivalence}
The set rewriting system $(P,R)$ is immortal if and only if there are infinitely many words not accepted by $\mathcal{A}^*$.
\end{lemma}
\begin{proof}
Suppose that the set rewriting system is immortal.
For every $k>0$, we will construct a non-accepted word $w$ of length at least $k\cdot(\ell+1)$.
Since the system is immortal and by Lemma~\ref{lem:set_rewriting_singleton}, there exists a singleton $\{p_i\}$ and a legal sequence $r_{i_1},\ldots,r_{i_k}$ of $k$ rules for $\{p_i\}$.
Hence, $w = r_{i_1} \alpha^\ell \ldots r_{i_k} \alpha^\ell$ is a simulating word for $S=\{p_i\}$.
By Lemma~\ref{lem:fpfl-simulating_word}, we know that $\delta_{\mathcal{A}^*}(\{q_0\},\alpha^i w) \subseteq Q_\mathrm{P} \cup \{q_\mathrm{g}\}$, which does not contain any final states, thus $\alpha^i w$ is not accepted.

Conversely, suppose that $L^*$ is not cofinite.
Then there are infinitely many words that are not accepted, which, in particular, by Lemma~\ref{lem:fpfl-irrevocably_accepted}, must be $f_0$-omitting.
Let $w$ be a $f_0$-omitting word of length at least $\ell+(\ell+1)2^{|P|}$.
By Lemma~\ref{lem:fpfl-initial_must_be_simulating}, we know that $w$ has the form of $\alpha^i w'$, where $1 \le i \le \ell$ and $w'$ is a prefix of a simulating word for $\{p_i\}$.
This simulating word must have length at least $(\ell+1)2^{|P|}$, hence it contains a sequence of $k \ge 2^{|P|}$ rule letters.
We conclude that this sequence $r_{i_1},r_{i_2},\ldots,r_{i_k}$ is legal for $\{p_i\}$, and it does not lead to the empty set as the set rewriting system is non-emptying.
If we look at the sequence of the sets of elements $S_j = \{p_i\} \cdot r_{i_1} \cdot \ldots \cdot r_{i_j}$, for $j \in \{0,\ldots,2^{|P|}\}$, then we can find two distinct indices $x$ and $y$ such that $x < y$ and $S_x = S_y$.
Hence, the rewriting system is immortal due to $S_x$ and the sequence $r_{i_{x+1}}, r_{i_{x+2}}, \ldots, r_{i_y}$.
\end{proof}

We conclude this part with
\begin{theorem}\label{thm:fpfl-dfa}
Problem~\ref{pbm:fpfl} is PSPACE-hard if $L$ is specified by a DFA over a given alphabet.
\end{theorem}

\subsection{Binarization}\label{subsec:fpfl-binarization}

We show that the PSPACE-hardness still holds when the alphabet is restricted to two letters.
Note that a standard binarization of an arbitrary language, where we uniformly replace each letter with an equal-length binary encoding, does not work here.
Except for some trivial cases, if incomplete encodings are accepted, then the Kleene star will be always cofinite, and if they are not accepted, then it will never be cofinite.
Therefore, we have to use a different way and utilize specific properties of the construction.

First, we define the following binary encoding $\mathit{bin}\colon \Sigma \to \{0,1\}^*$:
let $\mathit{bin}(\alpha)=0$, $\mathit{bin}(r_i) = 1^i 0$ for all $0 \leq i \leq m-1$, and $\mathit{bin}(r_m) = 1^m$.
We extend the function $\mathit{bin}$ to a function $\mathit{bin}\colon \Sigma^* \to \{0,1\}^*$ in a natural way.
Note that our encoding is a maximal prefix code, which means that $\mathit{bin}(u) \neq \mathit{bin}(v)$ for $u \neq v$, and also, every binary word $w' \in \{0,1\}^*$ contains a unique maximal prefix that is the encoding of some word over $\Sigma$; this prefix has length at least $|w'|-(m-1)$.

\begin{figure}[htb]\begin{center}
\resizebox{12cm}{!}{\large
\tikzset{every state/.style={minimum size=1.35cm}}
\tikzset{every path/.style={-{Latex[width=1mm,length=1.75mm]}}}

\tikzstyle{accepting}=[path picture={
\draw let 
  \p1 = (path picture bounding box.east),
  \p2 = (path picture bounding box.center)
  in
    (\p2) circle (\x1 - \x2 - 2pt);
 }]
\begin{tikzpicture}

\node[state, inner sep=0pt] (p_2) {$p_i$};
\node[state, inner sep=0pt, draw=none] (p_1) [left= 5.0 of p_2] {$\ldots$};
\node[state, inner sep=0pt, draw=none] (p_left_space) [left=of p_2] {};
\node[state, inner sep=0pt, draw=none] (p_3) [right= 5.0 of p_2] {$\ldots$};
\path[]
(p_1) edge node [above] {$0$} (p_2)
(p_2) edge node [above] {$0$} (p_3)
;
\node[state, inner sep=0pt, draw=none] (p_0) [left= of p_left_space] {};
\node[state, inner sep=0pt] (c_0_1) [below=0.8 of p_0] {$c^{i,1}$};
\node[state, inner sep=0pt,draw=none] (c_0_2) [right= of c_0_1]
{$\ldots$};
\node[state, inner sep=0pt] (c_0_3) [right= of c_0_2] {$c^{i,m-2}$};
\node[state, inner sep=0pt] (c_0_4) [right= of c_0_3] {$c^{i,m-1}$};
\path[]
(p_2) edge node [above left] {$1$} (c_0_1)
(c_0_1) edge node [above] {$1$} (c_0_2)
(c_0_2) edge node [above] {$1$} (c_0_3)
(c_0_3) edge node [above] {$1$} (c_0_4)
;

\node[state, inner sep=0pt] (s_0_1) [below =0.8 of c_0_1] {$s^{i,1}_1$};
\node[state, inner sep=0pt] (s_0_3) [below =0.8 of c_0_3] {$s^{i,m-2}_1$};
\node[state, inner sep=0pt] (s_0_4) [below =0.8 of c_0_4] {$s^{i,m-1}_1$};
\node[state, inner sep=0pt] (s_0_5) [right =of s_0_4] {$s^{i,m}_1$};

\node[state,inner sep=0pt,draw=none] (s_1_1) [below = 0.8 of s_0_1] {$\vdots$};
\node[state,inner sep=0pt,draw=none] (s_1_3) [below = 0.8 of s_0_3] {$\vdots$};
\node[state,inner sep=0pt,draw=none] (s_1_4) [below = 0.8 of s_0_4] {$\vdots$};
\node[state,inner sep=0pt,draw=none] (s_1_5) [below = 0.8 of s_0_5] {$\vdots$};

\path[]
(s_0_1) edge node [left] {$0$} (s_1_1)
(s_0_3) edge node [left] {$0$} (s_1_3)
(s_0_4) edge node [left] {$0$} (s_1_4)
(s_0_5) edge node [left] {$0$} (s_1_5)
;

\path[]
(c_0_1) edge node [left] {$0$} (s_0_1)
(c_0_3) edge node [left] {$0$} (s_0_3)
(c_0_4) edge node [left] {$0$} (s_0_4)
(c_0_4) edge node [above right] {$1$} (s_0_5)
;
\end{tikzpicture}}
\end{center}\caption{The fragment of the binary DFA $\mathcal{B}$ with the choice states of a state $p_i$. The transitions on $1$ from the setting states are not drawn and go to $f_0$.}\label{fig:fpfl_DFA-binary}
\end{figure}

We modify the construction of $\mathcal{A}$ from Subsection~\ref{subsec:fpfl-dfa} using the same notation.
We construct a binary DFA $\mathcal{B}=(Q_\mathcal{B},\{0,1\},\delta_\mathcal{B},q_0,F_\mathcal{B})$, where $Q_\mathcal{B}$ is $Q_\mathcal{A}$ with some states added, and $q_0$ and the set of final states $F_\mathcal{B}=F_\mathcal{A}$ are the same as in the original $\mathcal{A}$.
All the transitions labeled by $\alpha$ are now labeled by $0$.
For each state $p_i \in Q_\mathrm{P}$, we introduce $m-1$ new intermediate \emph{choice states} in the way that the binary word encoding $\mathit{bin}(r_j)$ of a rule letter $r_j$ acts as $r_j$ on $p_i$ in $\mathcal{A}$.
The construction of these states is shown in Fig.~\ref{fig:fpfl_DFA-binary}.
Formally, we add states $c^{i,j}$ for all $i \in \{1,\ldots,\ell\}$ and $j \in \{1,\ldots,m-1\}$, and the related transitions for all $i$ are defined by:
\begin{itemize}
\item $\delta_\mathcal{B}(p_i,1) = c^{i,1}$.
\item $\delta_\mathcal{B}(c^{i,j},1) = c^{i,j+1}$ for all $j \in \{1,\ldots,m-2\}$.
\item $\delta_\mathcal{B}(c^{i,j},0) = \begin{cases}
s^{i,j}_1, & \text{if}\ r_j(p_i) \neq \bot\\
f_0, & \text{otherwise}
\end{cases}$\\
for all $j \in \{1,\ldots,m-1\}$.
\item $\delta_\mathcal{B}(c^{i,m-1},1) = \begin{cases}
s^{i,m}_1, & \text{if}\ r_m(p_i) \neq \bot\\
f_0, & \text{otherwise}.
\end{cases}$
\end{itemize}

The transitions of the rule letters on $Q_\mathcal{B} \setminus Q_\mathrm{P}$ are simply replaced with one transition labeled by $1$; all of those transitions in $\mathcal{A}$ are the same for every $r_j \in R$ and lead to either $f_0$ or $q_\mathrm{s}$.

The correctness of the binarization is observed through the following lemmas.
First, we state that for a word that is $f_0$-omitting in the original automaton, the corresponding binary encoding works the same in the binarized one.
Then we have a similar statement for not $f_0$-omitting words but restricted to keeping the property of being $f_0$-omitting.

\begin{lemma}\label{lem:fpfl-binary_f0-omitting}
If a word $w \in \Sigma^*$ is $f_0$-omitting for a subset $C \subseteq Q_{\mathcal{A}^*} \setminus Q_\mathrm{F}$ in $\mathcal{A}^*$, then $\delta_{\mathcal{B}^*}(C,\mathit{bin}(w)) = \delta_{\mathcal{A}^*}(C,w)$.
\end{lemma}
\begin{proof}
This can be observed by analyzing the transitions of each letter in $\Sigma$ from each state in $Q_{\mathcal{A}^*} \setminus Q_\mathrm{F}$ in $\mathcal{A}^*$ together with the actions of the corresponding binary encodings in $\mathcal{B}^*$.
\end{proof}

From Lemma~\ref{lem:fpfl-binary_f0-omitting}, in particular, if a word $w \in \Sigma^*$ is $f_0$-omitting for a subset $C$ in $\mathcal{A}$, then $\mathit{bin}(w)$ is also $f_0$-omitting for $C$ in $\mathcal{B}^*$.

\begin{lemma}\label{lem:fpfl-binary_not_f0-omitting}
If a word $w \in \Sigma^*$ is not $f_0$-omitting for a subset $C \subseteq Q_{\mathcal{A}^*} \setminus Q_\mathrm{F}$ in $\mathcal{A}^*$, then $\mathit{bin}(w)$ is not $f_0$-omitting for $C$ in $\mathcal{B}^*$.
\end{lemma}
\begin{proof}
Suppose that a prefix of $w$ activates $f_0$ when applied to $\{q_0\}$; let $ua$ be a shortest such a prefix, where $u \in \Sigma^*$ and $a \in \Sigma$.
Since $u$ is $f_0$-omitting, from Lemma~\ref{lem:fpfl-binary_f0-omitting}, we know that $T = \delta_{\mathcal{A}^*}(C,u) = \delta_{\mathcal{B}^*}(C,\mathit{bin}(u))$.
If $a = \alpha$, then $0$ applied to $T$ activates $f_0$ in $\mathcal{B}^*$, as $\alpha$ does it in $\mathcal{A}^*$.
If $a \in R$, then in $\mathcal{A}^*$, we can have $q_0 \in T$, a state $s^{i,j}_k \in T$, or a state $p_i \in T$ mapped to $f_0$ by the transition of $a$.
In the first two cases, in $\mathcal{B}^*$, letter $1$ activates $f_0$ from $T$, and in the third case, $\mathit{bin}(u)\mathit{bin}(a)$ activates $f_0$ from $T$.
Both $\mathit{bin}(u)1$ and $\mathit{bin}(u)\mathit{bin}(a)$ are prefixes of $\mathit{bin}(w)$.
\end{proof}

From Lemma~\ref{lem:fpfl-binary_f0-omitting}, one direction of the cofiniteness equivalence follows easily.
For the second, we still need to consider binary words that are not complete encodings of words over the original alphabet.
We also need to observe that Lemma~\ref{lem:fpfl-irrevocably_accepted} holds for $\mathcal{B}^*$ as well.

\begin{lemma}\label{lem:fpfl-binary_equivalence}
The language of $\mathcal{B}^*$ is cofinite if and only the language of $\mathcal{A}^*$ is cofinite.
\end{lemma}
\begin{proof}
From Lemma~\ref{lem:fpfl-binary_f0-omitting} and by the fact that all not $f_0$-omitting words are accepted by $\mathcal{A}^*$, we know that if a word $w \in \Sigma^*$ is not accepted by $\mathcal{A}^*$, then $\mathit{bin}(w)$ is not accepted by $\mathcal{B}^*$.
Thus, if infinitely many words are not accepted by $\mathcal{A}^*$, then the language of $\mathcal{B}^*$ is also not cofinite.

Assume now that the language of $\mathcal{B}^*$ is not cofinite.
For an integer $t \ge m$, let $w' \in \{0,1\}^*$ be a binary word not accepted by $\mathcal{B}^*$ that has length at least $t$.
Let $u'$ be the longest prefix of $w'$ that properly encodes a word $u \in \Sigma^*$, i.e., $\mathit{bin}(u)=u'$; then $u'$ is shorter by at most $m-1$ letters than $w'$.
Observe that Lemma~\ref{lem:fpfl-irrevocably_accepted} holds for $\mathcal{B}^*$; hence, since $w'$ is not accepted, $u'$ must be $f_0$-omitting.
From Lemma~\ref{lem:fpfl-binary_not_f0-omitting}, we know that $u$ also must be $f_0$-omitting.
The length of $u$ is at least $(t-m+1)/m$, since each letter from $\Sigma$ is encoded by at most $m$ letters from $\{0,1\}$.
Now we follow similarly as in the proof of Lemma~\ref{lem:fpfl-equivalence}.
From Lemma~\ref{lem:fpfl-initial_must_be_simulating}, we conclude that $u$ has a prefix of length at least $|u|-\ell$ that is not accepted.
Hence, for every length $k$, we can choose a suitable $t \ge m(k+\ell)+m-1$ to find a non-accepted word of length at least $k$, thus there are infinitely many non-accepted words by $\mathcal{A}^*$.
\end{proof}

This finishes the reduction to the case of a binary DFA.

\subsection{List of words}\label{subsec:fpfl-list_of_words}

Finally, we count the largest length and the number of words in the language accepted by $\mathcal{B}$.

\begin{lemma}\label{lem:fpfl-binary_numbers}
In the language of $\mathcal{B}$, the maximum length of words is equal to $3\ell+m+1$ and the number of words is at most $m \ell^2 + (1 + \ell m(2+\ell))(1+\ell) \in \varTheta(m^2 \ell^2)$.
\end{lemma}
\begin{proof}
The maximum length of words accepted by our binary DFA $\mathcal{B}$ is equal to 
$3\ell+m+1$, which is the length of the longest path from $q_0$ to a final state:
$q_0 \xrightarrow{0^{\ell}} p_\ell \xrightarrow{1^m} s^{\ell,m}_\ell \xrightarrow{0^{\ell-1}} s^{\ell,m}_1 \xrightarrow{0} q_\mathrm{g} \xrightarrow{0} f_0 \xrightarrow{0^{\ell}} f_\ell$.

For the number of words in the recognized language, we consider all final states.
The first type of final states is setting states.
Each setting state is reachable from $q_0$ by a unique word, which gives at most $m \ell^2$ words.
The second type is forcing states.
A forcing state $f_i$ is reachable by different words, but all such words have a prefix for reaching $f_0$ followed by the unique suffix $\alpha^i$ to map $f_0$ to $f_i$.
For the number of the first parts, observe that all the states in $Q_\mathcal{B} \setminus (Q_\mathrm{F} \cup \{q_\mathrm{g},q_\mathrm{s}\})$, whose number is at most $1+\ell m(1+\ell)$, are reachable by a unique word, and $q_\mathrm{g}$ is reachable by at most $m\ell$ words.
Also, from each of these states, only one transition leads directly to $f_0$, with the possible exception of the last choice states $c^{i,m-1}$ when both rules $r_{m-1}$ and $r_m$ are illegal for an element $p_i$; however, in this case, we count fewer words than when they are legal, as the chain of setting states does not exist.
Thus, combining with the second parts, there are at most $(1 + \ell m(2+\ell))(1+\ell)$ words of this type.
\end{proof}

We conclude with

\begin{theorem}\label{thm:fpfl-binary_dfa}
Problem~\ref{pbm:fpfl} is PSPACE-complete if $L$ is a finite list of binary words.
\end{theorem}

Using the construction, we can also infer the hardness for every fixed-sized alphabet larger than binary.
For this, it is sufficient to add a suitable number of additional letters to $\mathcal{B}$ with the transitions mapping $Q_\mathcal{B} \setminus (Q_\mathrm{F} \cup \{q_\mathrm{s}\})$ to $f_0$ and mapping $Q_\mathrm{F} \cup \{q_\mathrm{s}\}$ to $q_\mathrm{s}$, thus acting as rule letters that are always illegal in the set rewriting system.

\section{The factor universality problem}\label{sec:fu}

We follow similarly as in Section~\ref{sec:fpfl}.
We reduce from Problem~\ref{pbm:emptying_set_rewriting} (Emptying Set Rewriting) to Problem~\ref{pbm:factor_universality} (Factor Universality for a Finite Set of Words) when $L$ is given as a finite list of binary words.

\subsection{DFA construction}

In the first step, we reduce to Problem~\ref{pbm:factor_universality} when $L$ is specified as a DFA instead of a list of words.

We slightly modify the DFA construction $\mathcal{A}$ from Subsection~\ref{subsec:fpfl-dfa} as follows.
We remove the last forcing state $f_\ell$ and end the chain of the forcing states with $f_{\ell-1}$.
Thus, the set of forcing states $Q_\mathrm{F}$ becomes $\{f_i \mid i \in \{0,1,\ldots,\ell-1\}\}$, and we redefine the transition $\delta_\mathcal{A}(f_{\ell-1},\alpha) = q_\mathrm{s}$.
There are no other differences.

\subsubsection{The mechanism}

The idea of the modified construction is as follows.

As before, we build the NFA $\mathcal{A}^*$ recognizing the language $L^*$, where $L$ is the language of $\mathcal{A}$.
In the NFA $\mathcal{A}^*$, all states are reachable from the initial state $q_0$, and since we have also removed the sink state $q_\mathrm{s}$, the NFA meets the criterion for factor universality from Proposition~\ref{pro:fu_criterion}.
Thus the language of $\mathcal{A}^*$ is factor universal if and only if there is a $Q_{\mathcal{A}^*}$-emptying word.

Simulating words in our NFA correspond to applications of sequences of rules in the set rewriting system in the same way as before in Subsection~\ref{subsec:fpfl-dfa}.
In the modified construction, the forcing states have the property that whenever $f_0$ is activated, the only way to get rid of all active forcing states is to make the whole set $Q_\mathrm{P}$ active again.
When $f_0$ is active, this is done by applying the word $\alpha^\ell$.
In this way, the construction ensures that to map the whole set $Q_\mathrm{P}$ to the empty set, there must exist a simulating word whose sequences of rules is $P$-emptying in the set rewriting system.
A special case occurs at the beginning, where we start with all the states $Q_{\mathcal{A}^*}$ and first have to reduce the active set of states to $Q_\mathrm{P}$.

\subsubsection{Correctness}

The correctness is observed through the following lemmas.

Recall that a word $w \in \Sigma^*$ is \emph{$f_0$-omitting for a subset} $C \subseteq Q_{\mathcal{A}^*}$ if there is no prefix $u$ of $w$ such that $f_0 \in \delta_{\mathcal{A}^*}(C,u)$.
Since in this problem our starting set is $Q_\mathrm{P}$ instead of $\{q_0\}$, we redefine that a word is simply \emph{$f_0$-omitting} if it is $f_0$-omitting for $Q_\mathrm{P}$.

We start with a simple observation, which follows directly from the construction and allows reducing the problem of emptying the whole set of states to emptying $Q_\mathrm{P}$.

\begin{lemma}\label{lem:fu-beginning}
We have:
\begin{enumerate}
\item $\delta_{\mathcal{A}^*}(Q_{\mathcal{A}^*},r_j^2) = \{f_0,q_0\}$ for each $r_j \in R$, and
\item $\delta_{\mathcal{A}^*}(\{f_0\},\alpha^\ell) = Q_\mathrm{P}$.
\end{enumerate}
\end{lemma}

We show that when $f_0$ is activated, the only way to get rid of all forcing states is to activate the whole $Q_\mathrm{P}$ at some point.

\begin{lemma}\label{lem:fu-cleaning_f0}
Let $C \subseteq Q_{\mathcal{A}^*}$ be such that $f_0 \in C$, and let $w$ be a word such that $\delta_{\mathcal{A}^*}(C,w) \cap Q_\mathrm{F} = \emptyset$.
There exists a prefix $u$ of $w$ such that $Q_\mathrm{P} \subseteq \delta_{\mathcal{A}^*}(C,u)$.
\end{lemma}
\begin{proof}
It is enough to prove the lemma for $C=\{f_0\}$.
Let $w$ be a shortest word such that $\delta_{\mathcal{A}^*}(\{f_0\},w) \cap Q_\mathrm{F} = \emptyset$.
Hence, there is no non-empty prefix $u$ of $w$ such that $f_0 \in \delta_{\mathcal{A}^*}(\{f_0\},u)$, as otherwise also $\delta_{\mathcal{A}^*}(\{f_0\},u') \cap Q_\mathrm{F} = \emptyset$ where $uu'=w$, and $u'$ would be shorter than $w$.

Observe that $\delta_{\mathcal{A}^*}(\{f_0\},\alpha^i) = \{f_i,q_0,p_1,\ldots,p_i\}$ for all $0 \le i \le \ell-1$.
Thus, if $w$ would start with $\alpha^i r_j$ for $0 \le i \le \ell-1$ and some rule letter $r_j$, then the active state $q_0$ would be mapped to $f_0$ by the transition of $r_j$, which yields a contradiction.
The only remaining possibility is that $w$ begins with the prefix $u=\alpha^\ell$, which is such that $\delta_{\mathcal{A}^*}(\{f_0\},u)=Q_\mathrm{P}$.
\end{proof}

We show the properties of a simulating word, in particular, that they correspond to rule applications in the set rewriting system.
A special case occurs when we reach the empty set; then we do not activate the guard state.

\begin{lemma}\label{lem:fu-simulating_word}
Let $C \subseteq Q_\mathrm{P} \cup \{q_\mathrm{g}\}$ be such that $S = C \cap Q_\mathrm{P}$ is non-empty, and let $w = r_{i_1} \alpha^\ell \ldots r_{i_k} \alpha^\ell$ ($k \ge 1$) be a simulating word for $S$. Then $w$ is $f_0$-omitting and
\[ \delta_{\mathcal{A}^*}(C,w) = \begin{cases}
  (S \cdot r_{i_1}\cdot\ldots\cdot r_{i_k}) \cup \{q_\mathrm{g}\}, & \text{if}\ S \cdot r_{i_1}\cdot\ldots\cdot r_{i_{k-1}} \neq \emptyset \\
  \emptyset = (S \cdot r_{i_1}\cdot\ldots\cdot r_{i_k}), & \text{otherwise.}
\end{cases}\]
\end{lemma} 
\begin{proof}
In the case of $S \cdot r_{i_1}\cdot\ldots\cdot r_{i_{k-1}} \neq \emptyset$, the proof is the same as that of Lemma~\ref{lem:fpfl-simulating_word}, since for all $0 \le j \le k-1$, we have $S \cdot r_{i_1}\cdot\ldots\cdot r_j \neq \emptyset$, thus all preconditions apply.

Otherwise, let $j < k$ be the smallest index such that the set $S \cdot r_{i_1}\cdot\ldots\cdot r_{i_j}$ is empty.
By the argument for the first case, we know that $\delta_{\mathcal{A}^*}(S,r_{i_1}\alpha^\ell\ldots r_{i_j}\alpha^\ell) = \{q_\mathrm{g}\}$.
Applying the next letter $r_{i_{j+1}}$ removes this single state, yielding the empty set.
\end{proof}

For the other direction, words that are $f_0$-omitting must involve simulating words.
A special case occurs when we reach the empty set; then, there are no further restrictions, in particular, on that, we must continue with a simulating word.

\begin{lemma}\label{lem:fu-must_be_simulating}
Let $C = S \cup \{q_\mathrm{g}\}$, where $S \subseteq Q_\mathrm{P}$ is non-empty.
If a word $w$ is $f_0$-omitting for $C$, then either:
\begin{enumerate}
\item $w$ is a prefix of a simulating word for $S$, or
\item a prefix of $w$ is a simulating word for $S$ whose sequence of rules is $S$-emptying.
\end{enumerate}
\end{lemma}
\begin{proof}
Following the proof of Lemma~\ref{lem:fpfl-must_be_simulating}, we observe that the word $w$ must start with $r_j \alpha^\ell$, unless it ends prematurely.
Then, by Lemma~\ref{lem:fu-simulating_word}, we have $\delta_{\mathcal{A}^*}(C,r_j \alpha^\ell) = (S \cdot r_j) \cup \{q_\mathrm{g}\}$. 
We apply this argument iteratively for the obtained subset and the remainder of $w$, until either $w$ ends, in which case~(1) holds, or the set of resulting active states becomes $\{q_\mathrm{g}\}$, in which case~(2) holds.
\end{proof}

Using our ingredients collected so far, we can show that the existence of a $Q_\mathrm{P}$-emptying word implies the existence of a $P$-emptying sequence of rules.

\begin{lemma}\label{lem:fu-must_contain_emptying}
Let $w$ be a word such that $\delta_{\mathcal{A}^*}(Q_\mathrm{P},w) = \emptyset$.
Then $w$ contains a factor $v$ which is a simulating word for $Q_\mathrm{P}$ whose sequence of rules is $P$-emptying.
\end{lemma}
\begin{proof}
It is enough to prove the lemma for words $w$ that do not have a non-empty prefix $u$ such that $Q_\mathrm{P} \subseteq \delta_{\mathcal{A}^*}(Q_\mathrm{P},u)$; otherwise, we can search for a factor $v$ in $w$ with $u$ removed.
Hence, by Lemma~\ref{lem:fu-cleaning_f0}, $w$ must be $f_0$-omitting.
By Lemma~\ref{lem:fu-must_be_simulating}, we have two possibilities (1) and~(2).
In case~(2), we immediately know that $w$ contains a prefix that is a simulating word for $Q_\mathrm{P}$ whose sequence of rules is $P$-emptying.
In case~(1), $w$ is a prefix of a simulating word for $Q_\mathrm{P}$.
If $w$ itself is a simulating word, let $v=w$; otherwise, write $w = v r_{i_{k+1}} \alpha^i$ for a simulating word $v = r_{i_1} \alpha^\ell \ldots r_{i_k} \alpha^\ell$ for $Q_\mathrm{P}$ ($k \ge 0$) and some $0 \le i < \ell$.
Let $C' = \delta_{\mathcal{A}^*}(Q_\mathrm{P},v)$.
By Lemma~\ref{lem:fu-simulating_word}, $C' \subseteq Q_\mathrm{P} \cup \{q_\mathrm{g}\}$ and $S' = C' \cap Q_\mathrm{P} = P\cdot r_{i_1}\cdot\ldots\cdot r_{i_k}$.
If $S' \neq \emptyset$, then the action of the possibly remaining suffix $r_{i_{k+1}} \alpha^i$ do not map $S'$ to $\emptyset$, which yields a contradiction with the assumption about $w$.
Therefore, $S' = \emptyset$, thus the sequence of rules in $v$ is $P$-emptying.
\end{proof}

Finally, we combine all the facts to show the equivalence between the reduced problems.

\begin{lemma}\label{lem:fu-equivalence}
The following conditions are equivalent:
\begin{enumerate}
\item The permissive set rewriting system $(P,R)$ admits a $P$-emptying sequence of rules.
\item There exists a $Q_\mathrm{P}$-emptying and $f_0$-omitting word for $\mathcal{A}^*$.
\item There exists a $Q_{\mathcal{A}^*}$-emptying word for $\mathcal{A}^*$.
\end{enumerate}
\end{lemma}
\begin{proof}
\noindent (1)~$\Rightarrow$~(2): Suppose that for the set rewriting system there is a sequence of rules $r_{i_1},\ldots,r_{i_k}$ that is $P$-emptying.
We take the word $w = r_1 \alpha^\ell \ldots r_k \alpha^\ell$, which is a simulating word for $Q_\mathrm{P}$.
By Lemma~\ref{lem:fu-simulating_word}, we conclude that $\delta_{\mathcal{A}^*}(Q_\mathrm{P},w) \subseteq \{q_\mathrm{g}\}$.
Thus, $w r_1$ is a $Q_\mathrm{P}$-emptying and $f_0$-omitting word.
 
\noindent (2)~$\Rightarrow$~(3):
If $w$ is a $Q_\mathrm{P}$-emptying word, then, by Lemma~\ref{lem:fu-beginning}, $\delta_{\mathcal{A}^*}(Q_{\mathcal{A}^*},r_1^2 \alpha^\ell w)=\emptyset$.

\noindent (3)~$\Rightarrow$~(1):
If there exists a $Q_{\mathcal{A}^*}$-emptying word $w \in \Sigma^*$, then, in particular, $\delta_{\mathcal{A}^*}(Q_\mathrm{P},w) = \emptyset$.
By Lemma~\ref{lem:fu-must_contain_emptying}, $w$ contains a factor $v$ which is a simulating word for $Q_\mathrm{P}$ whose sequence of rules is $P$-emptying.
\end{proof}

We conclude this part with
\begin{theorem}\label{fu-dfa}
Problem~\ref{pbm:factor_universality} is PSPACE-hard if $L$ is specified by a DFA over a given alphabet.
\end{theorem}

\subsection{Binarization}\label{subsec:fu-binarization}

We construct a binary DFA $\mathcal{B}$ exactly in the same way as in Subsection~\ref{subsec:fpfl-binarization} and use the same notation and the same binary encoding $\mathit{bin}$.
Note that the criterion for the factor universality from Proposition~\ref{pro:fu_criterion} still holds for $\mathcal{B}^*$.

We observe that Lemma~\ref{lem:fpfl-binary_f0-omitting} and Lemma~\ref{lem:fpfl-binary_not_f0-omitting} hold also in this case; it is because both constructions of $\mathcal{B}^*$ differ only on the set $Q_\mathrm{F}$, whose transitions are irrelevant for the observations.
Also, Lemma~\ref{lem:fu-cleaning_f0} holds for $\mathcal{B}^*$ as it holds for $\mathcal{A}^*$:

\begin{lemma}\label{lem:fu-binary_cleaning_f0}
Let $C \subseteq Q_{\mathcal{B}^*}$ be such that $f_0 \in C$, and let $w \in \{0,1\}^*$ be a word such that $\delta_{\mathcal{B}^*}(C,w) \cap Q_\mathrm{F} = \emptyset$.
There exists a prefix $u$ of $w$ such that $Q_\mathrm{P} \subseteq \delta_{\mathcal{B}^*}(C,u)$.
\end{lemma}
\begin{proof}
The proof is the same as that of Lemma~\ref{lem:fu-cleaning_f0}: it is sufficient to replace each $\alpha$ with $0$ and each $r_j$ with $1$.
\end{proof}

Now we show the equivalence of the existence of $Q_\mathrm{P}$-emptying words in both $\mathcal{A}^*$ and $\mathcal{B}^*$.

\begin{lemma}\label{lem:fu-binary_correspondance}
There is a $Q_\mathrm{P}$-emptying and $f_0$-omitting word for $\mathcal{A}^*$ if and only if there is such a word for $\mathcal{B}^*$.
In particular, if $w' \in \{0,1\}^*$ is such a word for $\mathcal{B}^*$, then $w'0 = \mathit{bin}(w)$ for some word $w \in \Sigma^*$ with this property for $\mathcal{A}^*$.
\end{lemma}
\begin{proof}
Let $w$ be a $Q_\mathrm{P}$-emptying and $f_0$-omitting word for $\mathcal{A}^*$.
From Lemma~\ref{lem:fpfl-binary_f0-omitting}, we know that $\mathit{bin}(w)$ is $f_0$-omitting and such that $\delta_{\mathcal{B}^*}(Q_\mathrm{P},\mathit{bin}(w)) = \delta_{\mathcal{A}^*}(Q_\mathrm{P},w) = \emptyset$.

Conversely, assume that there is a $Q_\mathrm{P}$-emptying and $f_0$-omitting binary word $w'$ for $\mathcal{B}^*$.
Since $\delta_{\mathcal{B}^*}(Q_\mathrm{P},w')=\emptyset$, we know that $w'0$ has the same properties.
Furthermore, $w'0$ must be such that $w'0=\mathit{bin}(w)$ for some word $w \in \Sigma^*$ because $0$ can appear only at the end in $\mathit{bin}(a)$ for each $a \in \Sigma$.
Then, from Lemma~\ref{lem:fpfl-binary_not_f0-omitting}, $w$ must be $f_0$-omitting.
From Lemma~\ref{lem:fpfl-binary_f0-omitting}, we conclude that $w$ has to be also $Q_\mathrm{P}$-emptying as $w'0$ is.
\end{proof}

Finally, we show the equivalence of the existence of a $Q_{\mathcal{B}^*}$-emptying word and of a $Q_\mathrm{P}$-emptying word, which finishes the reduction to the case of a binary DFA.

\begin{lemma}\label{lem:fu-binary_contains_emptying}
For $\mathcal{B}^*$, there exists a $Q_\mathrm{P}$-emptying and $f_0$-omitting word if and only if there exists a $Q_{\mathcal{B}^*}$-emptying word.
In particular, a $Q_{\mathcal{B}^*}$-emptying word contains a factor that is $Q_\mathrm{P}$-emptying and $f_0$-omitting.
\end{lemma}
\begin{proof}
Assume that there is a $Q_\mathrm{P}$-emptying word $w$.
We have $\delta_{\mathcal{B}^*}(Q_{\mathcal{B}^*},1^{m+1})=\{f_0,q_0\}$ and $\delta_{\mathcal{B}^*}(\{f_0\},0^\ell)=Q_\mathrm{P}$.
Thus, $\delta_{\mathcal{B}^*}(Q_{\mathcal{B}^*}, 1^{m+1} 0^\ell w) = \emptyset$.

Conversely, let $w$ be a $Q_{\mathcal{B}^*}$-emptying word.
Let $u$ be the longest prefix of $w$ such that $Q_\mathrm{P} \subseteq \delta_{\mathcal{B}^*}(Q_{\mathcal{B}^*},u)$, and write $w = uv$.
By Lemma~\ref{lem:fu-binary_cleaning_f0}, $v$ has to be $f_0$-omitting, as otherwise $u$ would be longer.
Hence, $v$ is a $Q_\mathcal{P}$-emptying and $f_0$-omitting word, and it is a factor of $w$.
\end{proof}

\subsection{List of words}\label{subsec:fu-list_of_words}

\begin{lemma}\label{lem:fu-binary_numbers}
in the language of $\mathcal{B}$, the maximum length of words is equal to $3\ell+m$ and the number of words is at most $m \ell^2 + (1 + \ell m(2 + \ell)) \ell$.
\end{lemma}
\begin{proof}
We count the maximum length and the words as in the proof of Lemma~\ref{lem:fpfl-binary_numbers}, taking into account that the chain of forcing states is shorter by $1$.
\end{proof}

We conclude with

\begin{theorem}\label{thm:binary_dfa}
Problem~\ref{pbm:factor_universality} is PSPACE-complete when the alphabet is binary.
\end{theorem}

As for the Frobenius monoid problem, in the same way, by adding a suitable number of letters, it is possible to show the hardness for every fixed-sized alphabet larger than binary.

\section{Lower bounds}

\subsection{The longest omitted words}

It is known that for each odd integer $n \geq 5$, there exists a set of binary words $L$, which have length at most $n$, such that $L^*$ is cofinite and the longest words not in $L^*$ are of length $\varOmega(n^2 2^{\frac{n}{2}})$ \cite{KaoShallitXu2008FPFM}.
However, the constructed $L$ contains exponentially many words in $n$, thus a large lower bound in terms of the size of $L$ could not be inferred.

We show a subexponential lower bound in $|L|$ and $\norm{L}_1$ on the length of the longest words not in $L^*$ when $L^*$ is cofinite.
The idea is to construct a list of binary words from a mortal set rewriting system whose longest legal sequences of rules have an exponential length (Theorem~\ref{thm:set_rewriting_maximal_legal_length}).

\begin{theorem}\label{thm:fpfl-length_lower_bound}
There exists an infinite family of finite sets $L$ of binary words such that $L^*$ is cofinite and the longest words not in $L^*$ are of length at least
$\frac{\norm{L}_\infty - 1}{4} \cdot 2^{\frac{\norm{L}_\infty - 1}{4}}$, and this length is $2^{\varOmega(\sqrt[4]{|L|})}$ in terms of $|L|$ and $2^{\varOmega(\sqrt[5]{\norm{L}_1})}$ in terms of $\norm{L}_1$.
\end{theorem}
\begin{proof}
For an $n \ge 2$, from Theorem~\ref{thm:set_rewriting_maximal_legal_length}, we take the set rewriting system $(P,R)$ with $|P|=|R|=n$ and the subset $S$ meeting the bound $2^n-2$.
Then we use the construction from Section~\ref{sec:fpfl} to create a binary DFA $\mathcal{B}$ and its list of binary words $L$.
Since the set rewriting system is mortal, $L^*$ is cofinite.

From Lemma~\ref{lem:fpfl-binary_numbers} ($\ell=m=n$), the length of the longest words in $L$ is equal to $4n + 1$ and there are at most $n^3+(1+n^2(2+n))(1+n) = n^4 + 4n^3 + 2n^2 + n + 1$ words, thus $|L| \in \O(n^4)$ and $\norm{L}_1 \in \O(n^5)$.

We take a binary simulating word $w'$ with the longest possible legal sequence of rules for $S$, thus also for some singleton $S' \subseteq S$.
From Lemma~\ref{lem:fpfl-simulating_word} and Lemma~\ref{lem:fpfl-binary_f0-omitting}, we know that $0^i w' \notin L^*$, for $1 \le i \le n$ corresponding to the initial singleton $S'$.
For $n \ge 2$, we can lower bound the length of the binary encoding of each rule letter by $2$.
Since the longest possible legal sequence of rules has length $2^n-2$ and one rule application corresponds to at least $n+2$ letters (i.e., $\mathit{bin}(r_j) 0^n$ for a rule letter $r_j$), the length of the word $0^i w'$ is at least $(2^n-2)\cdot(n+2) + 1$.
For $n \ge 2$, we have $(2^n-2)\cdot(n+2) + 1 \ge 2^n \cdot n$.

Since $n = \frac{\norm{L}_\infty - 1}{4}$, $n \in \varOmega(\sqrt[4]{|L|})$, and $n \in \varOmega(\sqrt[5]{\norm{L}_1})$, respectively, the length of the word $0^i w'$ is as in the theorem.
\end{proof}

\subsection{The shortest incompletable words}

We show that when $L^*$ is not factor universal, the length of the shortest incompletable words can be exponential in $\norm{L}_\infty$ and subexponential in $|L|$ and $\norm{L}_1$.
The idea is to construct a list of binary words from a permissive set rewriting system whose shortest legal sequences of rules that are $P$-emptying are of exponential length (Theorem~\ref{thm:long_emptying_set_rewriting}).

\begin{theorem}\label{thm:fu-length_lower_bound}
There exists an infinite family of finite sets $L$ of binary words such that the shortest incompletable words are of length at least
$\frac{\norm{L}_\infty}{4} \cdot 2^{\frac{\norm{L}_\infty}{4}}$, and this length is $2^{\varOmega(\sqrt[4]{|L|})}$ in terms of $|L|$ and $2^{\varOmega(\sqrt[5]{\norm{L}_1})}$ in terms of $\norm{L}_1$.
\end{theorem}
\begin{proof}
For an $n \ge 2$, from Theorem~\ref{thm:long_emptying_set_rewriting}, we take the set rewriting system $(P,R)$ with $|P|=|R|=n$ where the shortest $P$-emptying rule sequences have length equal to $2^n-1$.
Then we apply the construction from Subsection~\ref{sec:fu} to create a binary DFA $\mathcal{B}$ and its list of binary words $L$.
Since there exists a $P$-emptying sequence of rules, we know that there exists a $Q_{\mathcal{B}^*}$-emptying word in $\mathcal{B}^*$, thus $L^*$ is not factor universal.

We show a lower bound on the length of the shortest incompletable words.
Let $w' \in \{0,1\}^*$ be an incompletable word.
From the criterion from Proposition~\ref{pro:fu_criterion}, $w'$ is also $Q_{\mathcal{B}^*}$-emptying in $\mathcal{B}^*$.
From Lemma~\ref{lem:fu-binary_contains_emptying}, we know that $w'$ contains a factor $u'$ that is $Q_\mathrm{P}$-emptying and $f_0$-omitting for $Q_\mathrm{P}$.
From Lemma~\ref{lem:fu-binary_correspondance}, we know that the word $u \in \Sigma^*$ such that $\mathit{bin}(u)=u'0$ is $Q_\mathrm{P}$-emptying in $\mathcal{A}^*$.
By Lemma~\ref{lem:fu-must_contain_emptying}, $u$ contains as a factor a simulating word $v$ whose sequence of rules is $P$-emptying.
Since the shortest such a sequence of rules has length $2^n-1$, the word $v$ and also $u$ have length at least $(2^n-1)\cdot(n+2)$.
Moreover, both these words contain at least $(2^n-1)$ rule letters, as we have taken such a set rewriting system.
Since, for $n \ge 2$, each rule letter is encoded by at least two binary symbols, we conclude that $u'$, thus also $w'$, has length at least $(2^n-1) \cdot (n+2)-1$.
We have $(2^n-1) \cdot (n+2)-1 \ge 2^n \cdot n$.

Since $n = \frac{\norm{L}_\infty}{4}$, $n \in \varOmega(\sqrt[4]{|L|})$, and $n \in \varOmega(\sqrt[5]{\norm{L}_1})$, respectively, the length of every $Q_\mathrm{P}$-emptying word is as in the theorem.
\end{proof}

\section{Upper bounds}\label{sec:upper_bounds}

We show algorithms and upper bounds on the related lengths for both problems, which are exponential only in $\norm{L}_\infty$ while remain polynomial in $|L|$ thus also in $\norm{L}_1$.

For the Frobenius monoid problem, the upper bound $\frac{2}{2|\Sigma|-1}(2^{\norm{L}_\infty} |\Sigma|^{\norm{L}_\infty}-1)$ on the length of the longest words not in $L^*$ when $L^*$ is cofinite was known \cite[Theorem~6.1]{KaoShallitXu2008FPFM}.
We show an upper bound that involves both $\norm{L}_\infty$ and $|L|$.

\begin{theorem}\label{thm:fpfl_exptime}
Problem~\ref{pbm:fpfl} can be solved in time exponential in $\norm{L}_\infty$ while polynomial in $|L|$.
If $L^*$ is cofinite, then the longest words not in $L^*$ have length at most $|L|\cdot(2^{\norm{L}_\infty}-1)-1$.
\end{theorem}
\begin{proof}
The proof uses a similar idea to that in~\cite[Theorem~6.1]{KaoShallitXu2008FPFM} (\cite[Theorem~3.2.5]{Xu2009FPFM}), but involves the number of words $|L|$.
We can assume that $\varepsilon \notin L$ and $L \neq \emptyset$.

We construct a DFA $\mathcal{A}=(Q_\mathcal{A},\Sigma,\delta_\mathcal{A},q_\varepsilon,F_\mathcal{A})$ recognizing $L$ in a standard way that it forms a tree.
Thus $Q_\mathcal{A} = \{q_u \mid u\text{ is a prefix of a word in }L\} \cup \{q_\mathrm{s}\}$, where $q_\mathrm{s}$ is the unique dead state.
For $q_u \in Q_\mathcal{A}$ and $a \in \Sigma$, we define $\delta_\mathcal{A}(q_u,a)=q_{ua}$ if $ua$ is a prefix of a word in $L$ and $\delta_\mathcal{A}(q_u,a)=q_\mathrm{s}$, otherwise.
Then each word $w \in L$ has the action mapping the initial state $q_0$ to a distinct state.

By the standard construction for the Kleene star (Section~\ref{sec:preliminaries}), we construct an NFA $\mathcal{A}^*=(Q_{\mathcal{A}^*},\Sigma,\delta_{\mathcal{A}^*},q_\varepsilon,\{q_\varepsilon\})$ recognizing $L^*$.
Hence $\mathcal{A}^*$ is $\mathcal{A}$ with added an $\varepsilon$-transition from every final state to $q_\varepsilon$ and with the dead state $q_\mathrm{s}$ removed.
For a word $w \in \Sigma^*$, let the set of \emph{active} states be $\delta_{\mathcal{A}^*}(\{q_\varepsilon\},w)$.

We are going to observe that for every word $w$, there are no more than $|Q_{\mathcal{A}^*}|\cdot 2^{\norm{L}_\infty}+1$ active states.
We define the \emph{level} of a state $q_u \in Q_{\mathcal{A}^*}$ to be the length of $u$.
For every state $q_u$ and a letter $a$, the set $\delta_{\mathcal{A}^*}(\{q_u\},a)$ contains $q_{ua}$, if this state exists, and $q_\varepsilon$, if $q_{ua}$ was final in the DFA.
Hence, for a subset $C \subseteq Q_{\mathcal{A}*}$ with at most one state for each level, the transition of every letter, thus also the action of every word, preserves this property.
Since we start with $\{q_\varepsilon\}$, after reading any word for every level at most one state can be active.
Moreover, if a state $q_u$ is the active state with the largest level, then the set of all the states with smaller levels that can be active is determined, as they are the states $q_{u'}$ with $u'$ being a prefix of $u$.
We call them the \emph{possibly active states} of $q_u$.
Hence, the number of reachable subsets from $\{q_0\}$ is bounded by the number of the choices for the set of possibly active states and for its subset with actually active states.

Note that for a state $q_u$ where $u$ is a proper prefix of some word $w \in L$, the set of possibly active states of $q_u$ is contained in that set of $q_w$.
Thus it is sufficient to count only the sets of the possibly active states of $q_u$ with $u \in L$.

Also, the initial state $q_\varepsilon$ is active if and only if a final state in the DFA is active, with the exception of the initial subset $\{q_\varepsilon\}$.
Altogether, we have at most $|L|$ choices for the set of possibly active states combined with at most $2^{\norm{L}_\infty}-1$ choices for the non-empty subset of actually active states.
Additionally, there are the empty set and the initial singleton.
We obtain the upper bound $2+|L|\cdot(2^{\norm{L}_\infty}-1)$ on the number of reachable subsets from $\{q_\varepsilon\}$.

The problem of whether $\mathcal{A}^*$ recognizes a non-cofinite language is equivalent to whether in the space of reachable subsets there exists a cycle such that a subset without the unique final state $q_\varepsilon$ is reachable from it.
Thus, we can check this in time exponential in $\norm{L}_\infty$ and polynomial in $|L|$.

If the language is cofinite, the empty subset is not reachable (which was counted in the upper bound), and there exists a reachable cycle from which we can reach only subsets with the final state.
Thus, the longest words not in $L^*$ have length at most $|L|\cdot(2^{\norm{L}_\infty}-1)-1$.
\end{proof}

For the factor universality problem, only the trivial upper bound $2^{\norm{L}_1-\norm{L}_\infty+1}$ was known \cite{GusevPribavkina2011}.
Note that it is doubly-exponential if represented only in terms of $\norm{L}_\infty$.

\begin{theorem}\label{thm:factor_universality_exptime}
Problem~\ref{pbm:factor_universality} can be solved in time exponential in $\norm{L}_\infty$ while polynomial in $|L|$.
If the set $L \neq \emptyset$ is not complete, then the shortest incompletable words have length at most $\norm{L}_\infty+|L|\cdot 2^{\norm{L}_\infty}$.
\end{theorem}
\begin{proof}
The statement is trivial when $L=\{\varepsilon\}$, and we can assume $\varepsilon \notin L$.
We construct a DFA $\mathcal{A}$ and an NFA $\mathcal{A}^*$ for $L^*$ as in the proof of Theorem~\ref{thm:fpfl_exptime}.
The language $L^*$ is factor universal if and only if there exists a $Q_\mathcal{A^*}$-emptying word (Proposition~\ref{pro:fu_criterion}).

We follow similarly as in the proof of Theorem~\ref{thm:fpfl_exptime}, obtaining an upper bound on the set of reachable subsets from $Q_{\mathcal{A}^*}$.
Note that for every word $w$ of length at least $\norm{L}_\infty$, in $\delta_{\mathcal{A}^*}(Q_{\mathcal{A}^*},w)$ there is at most one state for each level.
Thus, when restricted to such words, there are at most $1+|L|\cdot(2^{\norm{L}_\infty}-1)$ reachable subsets (not counting $\{q_\varepsilon\}$ this time, since it is not reachable from $Q_{\mathcal{A}^*}$ as long as $L \nsubseteq \{\varepsilon\}$).
Since we start from $Q_{\mathcal{A}^*}$, at the beginning there could be more reachable subsets by words shorter than $\norm{L}_\infty$.

If there exists a $Q_{\mathcal{A}^*}$-emptying word $w$, then for every word $u$, the word $uw$ is also $Q_{\mathcal{A}^*}$-emptying.
Hence, to solve the problem, we can start from an arbitrary word $u$ of length $\norm{L}_\infty$, and then check the reachability of the empty set.
The length of the shortest $Q_{\mathcal{A}^*}$-emptying words is at most $\norm{L}_\infty+|L|\cdot(2^{\norm{L}_\infty}-1)$.
\end{proof}

Under a fixed-sized alphabet (as otherwise $\norm{L}_1$ can be arbitrarily large with respect to $\norm{L}_\infty$), we have $|L| \le |\Sigma|^{\norm{L}_\infty}$.
We conclude that $2^{\O(\norm{L}_\infty)}$ is a tight upper bound on the lengths related to both problems.

\section*{Acknowledgments}

We thank Amir M. Ben-Amram for the idea of a simpler way for proving Theorem~\ref{thm:immortality_set_rewriting}.
We also thank all the anonymous reviewers for their comments.
This work was supported by the National Science Centre, Poland under project number 2017/25/B/ST6/01920.

\bibliographystyle{plainurl}
\bibliography{frobenius_and_factor_universality}
\section*{Appendix}

\subsection*{Large length of the shortest incompletable words}

We define explicitly the family from the proof Theorem~\ref{thm:fpfl-length_lower_bound} of sets of words $L$ for which the shortest incompletable words in $L^*$ are of exponential length $\frac{\norm{L}_\infty}{4} \cdot 2^{\frac{\norm{L}_\infty}{4}}$ in terms of $\norm{L}_\infty$ and subexponential length $2^{\varOmega(\sqrt[5]{\norm{L}_1})}$ in terms of $\norm{L}_1$.

For a given $n \ge 2$, the words in $L$ are as follows.
The paths in the construction from the initial state to a final state, which correspond to words in $L$, are also listed.
We rename the elements in the set $P = \{b_0,b_1,\ldots,b_{n-1}\}$ from the set rewriting system in the proof to the elements from $\{p_1, p_2, \ldots, p_n\}$ such that $b_i = p_{i+1}$ as in the reduction.
In this way, the construction keeps the property that if $s^{i,j}_k$ is final and $p_i$ is active, then after $1^j 0 0^n$ (or $1^j 0^n$ if $j = n$), $p_k$ will be active.

The words coming from final states $f_x$ for $x \in \{0,1,\ldots,n-1\}$:
\begin{itemize}
\item $1 0^x$ for $x \in \{0,\ldots,n-1\}$;\quad ($q_0 \xrightarrow{1} f_0 \xrightarrow{0^x} f_x$)
\item $0^n 0 0^x$ for $x \in \{0,\ldots,n-1\}$;\quad ($q_0 \xrightarrow{0^n} p_n \xrightarrow{0} f_0 \xrightarrow{0^x} f_x$)
\item $0^i 1^j 0 0^k 1 0^x$ for $i \in \{1,\ldots,n\}$, $j \in \{1,\ldots, n-1\}$, $k \in \{0,\ldots,n-1\}$, and $x \in \{0,\ldots,n-1\}$;\quad ($q_0 \xrightarrow{0^i} p_i \xrightarrow{1^j 0} s^{i,j}_n \xrightarrow{0^k} s^{i,j}_{n-k} \xrightarrow{1} f_0 \xrightarrow{0^x} f_x$)
\item $0^i 1^n 0^k 1 0^x$ for $i \in \{1,\ldots,n\}$, $k \in \{0,\ldots,n-1\}$, and $x \in \{0,\ldots,n-1\}$;\quad ($q_0 \xrightarrow{0^i} p_i \xrightarrow{1^n} s^{i,n}_n \xrightarrow{0^k} s^{i,n}_{n-k} \xrightarrow{1} f_0 \xrightarrow{0^x} f_x$)
\item $0^i 1^j 0 0^n 0 0^x$ for $i \in \{1,\ldots,n\}$, $j \in \{1,\ldots, n-1\}$, and $x \in \{0,\ldots,n-1\}$;\quad ($q_0 \xrightarrow{0^i} p_i \xrightarrow{1^j 0} s^{i,j}_n \xrightarrow{0^n} q_\mathrm{g} \xrightarrow{0} f_0 \xrightarrow{0^x} f_x$)
\item $0^i 1^n 0^n 0 0^x$ for $i \in \{1,\ldots,n\}$ and $x \in \{0,\ldots,n-1\}$;\quad ($q_0 \xrightarrow{0^i} p_i \xrightarrow{1^n} s^{i,n}_n \xrightarrow{0^n} q_\mathrm{g} \xrightarrow{0} f_0 \xrightarrow{0^x} f_x$)
\end{itemize}

The words coming from the final setting states corresponding to the transition $r_j(p_j) = \{p_i \mid i \in \{0,1,2,\ldots,j-1\}\}$:
\begin{itemize}
\item $0^j 1^j 0 0^{n-k}$ for $j \in \{1,\ldots,n-1\}$ and $k \in \{1,\ldots,j-1\}$;\quad ($q_0 \xrightarrow{0^j} p_j \xrightarrow{1^j 0} s^{j,j}_n \xrightarrow{0^{n-k}} s^{j,j}_k$) 
\item $0^n 1^n 0^{n-k}$ for $k \in \{1,\ldots,n-1\}$;\quad ($q_0 \xrightarrow{0^n} p_n \xrightarrow{1^n} s^{n,n}_n \xrightarrow{0^{n-k}} s^{n,n}_k$)
\end{itemize}

The words coming from the final setting states corresponding to the transition $r_j(p_i) = P$ for $i \in \{0,1,2,\ldots,j-1\}$:
\begin{itemize}
\item $0^i 1^j 0 0^{n-k}$ for $j \in \{1,2,\ldots,n-1\}$, $i \in \{1,\ldots,j-1\}$, and $k \in \{1,2,\ldots,n\}$;\quad ($q_0 \xrightarrow{0^i} p_i \xrightarrow{1^j 0} s^{i,j}_n \xrightarrow{0^{n-k}} s^{i,j}_k$)
\item $0^i 1^n 0^{n-k}$ for $i \in \{1,\ldots,n-1\}$ and $k \in \{1,2,\ldots,n\}$;\quad ($q_0 \xrightarrow{0^i} p_i \xrightarrow{1^n} s^{n,n}_n \xrightarrow{0^{n-k}} s^{n,n}_k$)
\end{itemize}

The words coming from the final setting states corresponding to the transition $r_j(p_i) = \{p_i\}$ for $i \in \{j+1,j+2,\ldots,n-1\}$:
\begin{itemize}
\item $0^i 1^j 0 0^{n-i}$ for $j \in \{1,2,\ldots,n-1\}$ and $i \in \{j+1,\ldots,n\}$;\quad ($q_0 \xrightarrow{0^i} p_i \xrightarrow{1^j 0} s^{i,j}_n \xrightarrow{0^{n-i}} s^{i,j}_{i}$)
\end{itemize}

A program generating these examples is also available at~\cite{this-arxiv} as a source file.

\end{document}